\documentclass[compsoc, conference, letterpaper, 10pt, times]{IEEEtran}
\usepackage{textcomp}
\usepackage{amsmath}
\usepackage{amsthm}
\usepackage{amsfonts}
\usepackage{amssymb}
\usepackage[numbers]{natbib}
\usepackage{verbatim}
\usepackage{xspace}
\usepackage{graphicx}
\usepackage{psfrag}
\usepackage{url}
\usepackage{pifont}
\usepackage{subfigure}
\usepackage{adjustbox}
\usepackage{epstopdf}
\usepackage{algpseudocode}
\usepackage[boxed]{algorithm}
\usepackage{enumitem}
\usepackage[english]{babel}
\usepackage{tablefootnote} \usepackage{array} \usepackage[usenames,dvipsnames]{xcolor}
\usepackage{balance}
\usepackage{verbatim}
\usepackage[normalem]{ulem} \usepackage{mdwlist}
\usepackage{xspace}
\usepackage{tikz}
\usetikzlibrary{calc}
\usepackage{rotate}
\usepackage{multirow}
\usepackage{makecell} \usepackage{balance}

\def\centerhack#1{\hbox to 0pt{\hss\footnotesize #1\hss}}

\makeatletter
\def\@listi{\leftmargin\leftmargini
	\parsep 1\p@ \@plus0\p@ \@minus\p@
	\topsep 2\p@   \@plus0\p@ \@minus\p@
			\itemsep1\p@ \@plus0\p@ \@minus\p@}
\let\@listI\@listi\@listi
\makeatother

\def\centerhack#1{\hbox to 0pt{\hss\footnotesize #1\hss}}
\def\dchack#1{\vbox to 0pt{\vss{\hbox to 0pt{\hss#1\hss}}\vss}}

\renewcommand\paragraph[1]{\smallskip\noindent\textbf{#1.}}

\makeatletter
\newcommand{\StatexIndent}[1][3]{	\setlength\@tempdima{\algorithmicindent}	\Statex\hskip\dimexpr#1\@tempdima\relax}
\makeatother

\newcommand{\eg}{e.g.,\xspace}

\newcommand{\FS}{\ensuremath{\mathit{FS}}\xspace}

\newcommand{\concat}{\ensuremath{\mathbin{\|}}}
\newcommand{\name}{PHI\xspace}
\newcommand{\hornet}{HORNET\xspace}

\newcommand{\lap}{LAP\xspace}
\newcommand{\dovetail}{Dovetail\xspace}
\newcommand{\sphinx}{Sphinx\xspace}

\newcommand{\tor}{Tor\xspace}

\renewcommand{\hornet}{HORNET\xspace}

\newcommand{\chen}[1]{\ding{110}\ding{43}\textcolor{ForestGreen} {CC: {#1}}}
\newcommand{\david}[1]{\ding{43}\textcolor{magenta} {DB: {#1}}}
\newcommand{\george}[1]{\ding{43}\textcolor{blue} {GD: {#1}}}

\newcommand{\carmelasubscript}{\textsubscript{\textcolor{BurntOrange}{\textsf{
				\textbf{CT}}}}}
\newcommand{\carmela}[1]{\textcolor{BurntOrange}{\ding{110}\carmelasubscript~{
			#1}}}
\newcommand{\danielecolor}{Cerulean}
\newcommand{\danielesubscript}{\textsubscript{\textcolor{\danielecolor}{\textsf{\textbf{DA}}}}}
\newcommand{\daniele}[1]{\textcolor{\danielecolor}{\ding{110}\danielesubscript~{#1}}}
\newcommand{\dadelnoname}[1]{\bgroup\markoverwith{\textcolor{\danielecolor}{\rule[0.35ex]{2pt}{1pt}}}\ULon{#1}}
\newcommand{\dadel}[1]{\dadelnoname{#1}\kern0.1em\danielesubscript}

\newcommand{\dasugg}[1]{\textcolor{\danielecolor}{[#1]\danielesubscript}}
\newcommand{\dasubs}[2]{\dadelnoname{#1}\dasugg{#2}}

\renewcommand{\chen}[1]{}
\renewcommand{\david}[1]{}
\renewcommand{\george}[1]{}
\renewcommand{\carmela}[1]{}
\renewcommand{\daniele}[1]{}

\begin{document}
	
\newtheorem{theorem}{Theorem}[section]
\newtheorem{corollary}{Corollary}[theorem]
\newtheorem{definition}[theorem]{Definition}
\newtheorem{lemma}[theorem]{Lemma}	

\renewcommand{\name}{TARANET\xspace}
\newcommand{\nameacr}{Traffic-Analysis Resistant Anonymity \\at the NETwork layer}
\title{\name: \nameacr}
\author{\IEEEauthorblockN{Chen Chen}
	\IEEEauthorblockA{
		chenche1@andrew.cmu.edu\\ 
		Carnegie Mellon University
	} \\
	\IEEEauthorblockN{David Barrera}
	\IEEEauthorblockA{
		david.barrera@polymtl.ca \\
		Polytechnique Montreal
	}
	\and
	\IEEEauthorblockN{Daniele E. Asoni}
	\IEEEauthorblockA{
		daniele.asoni@inf.ethz.ch \\
		ETH Z\"urich
	} \\
	\IEEEauthorblockN{George Danezis}
	\IEEEauthorblockA{
	g.danezis@ucl.ac.uk \\
	University College London
	}
	\and
	\IEEEauthorblockN{Adrian Perrig}
	\IEEEauthorblockA{
		adrian.perrig@inf.ethz.ch \\
		ETH Z\"urich
	} \\
   \IEEEauthorblockN{Carmela Troncoso}
   \IEEEauthorblockA{
   	carmela.troncoso@epfl.ch \\
   EPFL
   }
}

\maketitle

\begin{abstract}
	Modern low-latency anonymity systems, no matter whether constructed as an 
overlay or implemented at the network layer, offer limited security guarantees
against traffic analysis. On the other hand, high-latency anonymity systems offer
strong security guarantees at the cost of computational overhead and long delays,
which are excessive for interactive applications. We propose \name, an anonymity system that implements protection
against traffic analysis at the network layer, and limits the incurred latency and overhead. In \name's
setup phase, traffic analysis is thwarted by mixing. In the data transmission phase,
end hosts and ASes coordinate to shape traffic into constant-rate transmission 
using packet splitting. Our prototype implementation shows that \name can forward anonymous traffic 
at over 50~Gbps using commodity hardware.

 \end{abstract}

\section{Introduction}
\label{sec:introduction}

Users are increasingly aware of their lack of privacy and are turning to anonymity
systems to protect their communications. \tor~\cite{dms04} 
is currently the most
popular anonymity system, with over 2 million daily users~\cite{torusers}.
Unfortunately, \tor offers neither satisfactory performance nor strong
anonymity. With respect to performance, \tor is implemented as an overlay
network and uses a per-hop reliable transport, increasing both propagation and
queuing latency~\cite{dingledine2009performance}. With respect to anonymity
guarantees, \tor is vulnerable to traffic analysis~\cite{murdoch2005low,overlier2006locating,murdoch2006hot,syverson2001towards}. 

Users also have the option of anonymity systems with stronger guarantees such as
DC-nets~\cite{chaum1988dining,golle2004dining,wolinsky2012dissent}, Mix
networks~\cite{chaum1981untraceable,berthold2001web}, and peer-to-peer
protocols~\cite{sherwood2002p, freedman2002tarzan}. However, these systems
either scale poorly or incur prohibitive latency and reliability, making them
unsuitable for many practical applications. 

In an effort to improve the performance of anonymity networks, research has
built on the idea of network-layer anonymity (e.g., \lap~\cite{Hsiao2012},
\dovetail~\cite{Sankey2014}, and \hornet~\cite{Chen2015Hornet}). Network-layer
anonymity systems assume that the network infrastructure (e.g., routers)
participates in establishing anonymous communication channels and assists in
forwarding anonymous traffic. Intermediate anonymity supporting network nodes
(or nodes for short) first cooperate with senders to establish anonymous
sessions or circuits, and then process and forward traffic from those senders to
receivers. While these systems achieve high throughput and low latency, the
security guarantees of these systems are no stronger than \tor{}'s. Moreover,
\lap and \dovetail leak the position of intermediate nodes on the path and the
total path length, which reduces the anonymity set size, facilitating
de-anonymization~\cite{Chen2015Hornet}.

The problem space appears to have an unavoidable tradeoff: \emph{strong anonymity appears achievable only through drastically higher overhead}~\cite{das2017anonymity}. In this paper, we aim to push the boundaries of this anonymity/performance tradeoff by combining the speed of network-layer anonymity systems with strong defenses. 

To improve the anonymity guarantees, traffic analysis attacks need to be
prevented, or made significantly harder/costlier to perform.
The common method to achieve this is to insert chaff, which are dummy packets
which to an adversary look indistinguishable from encrypted data packets.
By mixing chaff with data packets, one can add noise to the underlying
traffic patterns to defeat traffic analysis. For example, one can insert
chaff to maintain a constant transmission rate on an adversarial network
link, so that the traffic patterns observed by the observing adversary
stay unchanged and leak no identifying information.

However, both existing methods of applying chaff traffic, i.e.,
constant-transmission-rate link padding~\cite{wang2008dependent,
freedman2002tarzan,aqua,herd2015} and probabilistic end-to-end
padding~\cite{levine2004timing,piotrowska2017loopix}, are unsatisfactory. On one
hand, constant-transmission-rate link padding uses chaff to shape traffic
between adjacent pairs of nodes making it perfectly homogeneous, thus provably
concealing the underlying traffic patterns from a network adversary. However, a
compromised node is able to distinguish chaff traffic from real traffic, giving
link padding no anonymity guarantees when compromised nodes are present. On the
other hand, probabilistic end-to-end padding enables end hosts to generate chaff
traffic that is indistinguishable from real traffic, but existing
schemes~\cite{levine2004timing,piotrowska2017loopix} fail to fully conceal the
end-to-end transmission rate and can be defeated by packet-density
attack~\cite{shmatikov2006timing}. 

We take the best of both worlds and propose a new method of applying chaff traffic that
has so far not been explored:
an \emph{end-to-end} padding scheme that shapes a flow's traffic pattern into \emph{constant-rate transmission on all traversed links}. 
At a flow's origin, the sender divides its traffic
into small \emph{flowlets} that transmit packets at a globally-fixed constant
rate. Each forwarding node modulates the outgoing transmission rate of each
flowlet so that the transmission rate remains constant over time and also
remains constant across all links traversed by the flowlet. This approach prevents traffic
patterns from propagating across nodes. We
call this technique \emph{end-to-end traffic shaping}.

However, end-to-end traffic shaping is surprisingly tricky to achieve in the
presence of natural packet loss, adversarial packet drops, or packet propagation
delays. The main challenge for coordinated traffic shaping is how to maintain
constant-rate transmission across all traversing links when a forwarding node's incoming
transmission rate is lower than the outgoing transmission rate. A simple
approach that enables a forwarding nodes to create valid packets to send toward
the destination appears promising, but unfortunately, this approach could be
abused, as packet injection requires the cryptographic keys that the sender
shares with downstream nodes. Moreover, such an approach would enable two
malicious nodes that are on the same flowlet path to trivially link observed
packets of the same flowlet. Similarly, allowing a node to replay existing
packets cannot be permitted, as replicated packets themselves would constitute a
trivially detectable pattern.

An initial idea is to enable each node to have a spare packet queue, containing
packets that can be sent to make up for the difference between the incoming
transmission rate and the required outgoing transmission rate. But this poses a
conundrum: how can we fill up the spare packet buffer if the flowlet rate
remains constant in the first place? Our solution is \emph{packet splitting}, a
cryptographic mechanism which allows an end host to generate a packet that
splits into two different valid packets of the same size as the original packet
at a specific node. Through splittable packets, an end host can fill up the
spare packet queue at forwarding nodes, which in turn enables constant-rate
transmission even in case of lost or delayed incoming packets.

In this paper, we propose \name, a scalable, high-speed, and
traffic-analysis-resistant anonymous communication protocol, which uses the end-to-end 
traffic shaping assisted by packet splitting as one of its novel mechanisms.
\name is directly built into the network infrastructure to achieve short paths and high throughput. 
It uses mixing for its setup phase and end-to-end traffic shaping for
its data transmission phase to resist traffic analysis. Our paper makes the following contributions:

\begin{enumerate}[noitemsep,nolistsep]

\item We propose an efficient end-to-end traffic shaping technique that maintains
per-flow constant-rate transmission on all links and defeats traffic analysis attacks. 
We also propose in-network packet splitting as the enabling mechanism 
for the end-to-end traffic shaping technique.

\item We present an onion routing protocol that enables payload integrity protection, replay detection, and splittable packets, which are essential building blocks for end-to-end traffic shaping.

\item We design, implement, and evaluate the security and performance of \name.
Our prototype running on commodity hardware can forward over 50~Gbps of anonymous traffic,
showing the feasibility to deploy \name on high-speed links.

\end{enumerate}

\section{Background and Related Work}
\label{sec:background}
This section presents background on network-layer anonymity protocols. We also discuss adversarial traffic analysis techniques to de-anonymize end points, focusing on those that current network-layer anonymity protocols fail to deter.

\subsection{Network-layer Anonymity Protocols} \label{sec:networklayerprotocols}

Recent research~\cite{Hsiao2012,Sankey2014,Chen2015Hornet} proposes
\emph{network-layer anonymity systems} that incorporate anonymous communication as a service of network infrastructures in the Internet and next generation network architectures~\cite{Yang2007NIRA,godfrey2009pathlet,Xin2011SCION}.
The basic assumption of a network-layer anonymity system is that Autonomous Systems (AS) can conduct 
efficient cryptographic operations when forwarding packets to conceal forwarding
information. Additionally, a network-layer anonymity system uses direct
forwarding paths rather than reroute packets through overlay networks as in
Tor~\cite{dms04}. This processing would be done on (software) routers, for
instance, but more abstractedly the term \emph{node} is used to refer to
the device or set of devices dedicated to the anonymity system within an AS.

A network-layer anonymity system anonymizes its traffic by relying on ASes to collaboratively hide the forwarding paths between senders and receivers. 
We remark that a network-layer anonymity system can offer neither \emph{sender anonymity} nor \emph{recipient anonymity} as defined by  Pfizmann and K\"{o}hntopp~\cite{pfitzmann2001}. A compromised first-hop AS on the path can observe the sender of a message, violating sender anonymity. Similarly, a compromised last-hop AS can identify the receiver, which breaks recipient anonymity. Instead, a network-layer anonymity system offers 
\emph{relationship anonymity}~\cite{pfitzmann2001} that prevents linking two end hosts of a message.

\chen{David mentions that we could delete the following paragraph for additional space}
Besides anonymity, the basic design goals for a network-layer anonymity system are scalability and performance. With respect to scalability, a network-layer anonymity system minimizes the amount of state kept on network routers who possess limited high-speed memory. 
With respect to performance, a network-layer anonymity system should offer low latency and high throughput.

\textbf{\hornet~\cite{Chen2015Hornet}} improves on the security guarantees for
network-layer protocols by using full onion encryption to guarantee bitwise
unlinkability.
\hornet introduces several useful primitives for
stateless onion routing, which we extend in \name.

\hornet is circuit-based like overlay systems, but
it operates at the network layer.
As with LAP and Dovetail, processing data packets at intermediate nodes requires only
symmetric cryptography. This design comes at the expense of a relatively slow
round-trip time for setup packets which requires nodes on the path to perform public-key cryptography at the
start of each session. During setup, the sender establishes keys between itself
and every node on the path. The sender embeds these
keys along with routing information for each hop into the header of each
subsequent data packet. Since the state is carried within packets, intermediate
nodes do not have to keep per-flow state, which enables
high scalability.

Through bit-pattern unlinkability in its traffic and confidentiality of the packet's path information,
\hornet can defend against passive adversaries matching packets based on packet contents.
Nevertheless, the protocol is vulnerable to
more sophisticated active attacks.
\hornet headers are re-used for all data packets in a session, and
payloads are not integrity-protected.
Thus, \hornet cannot protect against packet replays since an
adversary could change a payload arbitrarily, making the packet look
indistinguishable from a legitimate new packet to the processing node.
Such a replay attack can be used in conjunction with traffic analysis to insert
recognizable fingerprints into flows, which can help de-anonymize communicating
endpoints.

\paragraph{Lightweight anonymity systems}
The first class of network-layer anonymity protocols proposed is
the so-called \emph{lightweight} system, which consists of two
proposals, \lap~\cite{Hsiao2012} and \dovetail~\cite{Sankey2014}.
These systems defend against topological attacks by encrypting forwarding information in
packet headers. However, in both schemes, packets stay unchanged from hop to
hop, thus enabling bit-pattern correlation of packets at distinct compromised
nodes.

\subsection{Traffic Analysis Attacks}
\label{sec:ta_techniques}

Traffic analysis aims to identify communicating endpoints based on
\emph{metadata} such as volume, traffic patterns, and timing. The literature
broadly classifies traffic analysis techniques into passive and active,
depending on whether the adversary manipulates traffic.

\subsubsection{Passive Attacks} \label{sec:ta_passive}

\paragraph{Flow dynamics matching}
An adversary eavesdropping on traffic at two observation points (including an
adversary observing the ingress and egress traffic of a single node) can try to
detect whether (some of) the packets seen at the observation points belong to
the same flow by searching for similarities among the dynamics of
all observed flows~\cite{Zhu2005,timing-fc2004,murdoch2005low,throughput2011}.
For example, the adversary can monitor packet inter-arrival times, flow
volume~\cite{blum2004detection}, or on/off flow
patterns~\cite{wang2003,zhang2000detecting}.

\paragraph{Template attacks}
An adversary can construct a database of traffic patterns (\emph{templates})
obtained by accessing known websites or other web-services through the anonymous
communication system. When eavesdropping on the traffic of a client, the
adversary compares the observed flows with the patterns stored in the database,
and if a match is found the adversary is able to guess the website or
web-service accessed by the client with high
probability~\cite{Gong2011,juarez2014,Wang2014}.

\paragraph{Network statistics correlation}
Another possible attack consists in monitoring network characteristics of
different parts of the network, and comparing them to the characteristics of
targeted anonymized flows. For instance, by comparing the round-trip time (RTT)
of a target bidirectional flow with the RTTs measured to a large set of
network locations, an adversary can identify the probable network location of
an end host in case the RTT of the flow showed strong correlation with the RTT
to one of the monitored network
locations~\cite{DBLP:journals/tissec/HopperVC10}. Similarly, by simply the
throughput (over time) of a unidirectional flow and comparing it with the
throughput to various network location, the adversary can guess the end host's
location~\cite{mittal2011stealthy}.\chen{I do not see why we group
the latency fingerprint attack and the throughput fingerprint attacks into 
one "network state correlation" attack. What about network-statistics fingerprinting attacks?}

\subsubsection{Active Attacks}
\label{sec:activeattacks}

Active traffic analysis uses similar techniques as passive traffic analysis, but
it additionally involves traffic manipulation by the adversary, in particular
packet delaying and dropping, to introduce specific patterns.
Chakravarty et al.~\cite{chakravarty2014} show that active
analysis can have high success rates even when working with aggregate Netflow
data instead of raw packet traces.

\paragraph{Flow dynamics modification}
By modifying the flow dynamics (inter-packet timings), the adversary can add a
\emph{watermark} (or \emph{tag}) to the flow, which the adversary is then able
to detect when observing the flow at another point in the
network~\cite{wang2003,houmansadr_rainbow:_2009,houmansadr_swirl:_2011}. This
attack is known as \emph{flow watermarking}. A similiar attack, called \emph{flow
fingerprinting}, enables an adversary to encode more
information into the flow dynamics, which can later be extracted from the same
flow seen at another point in the network~\cite{houmansadr2013}. For both
attacks, depending on the coding technique, flows may require more or fewer
packets for the watermark/fingerprint to be reliably identified within the
network.

\paragraph{Clogging Attacks}
Flow dynamics modification requires that the adversary
control multiple observation points in the network. Clogging attacks are
similar, but the adversary only needs to be able to observe the target flow at a
single network location. For these attacks, the adversary causes network
congestion~\cite{murdoch2005low, evans2009practical}, or
fluctuation~\cite{chakravarty2010traffic} at other nodes in the network, and
then observes whether these actions affect the observed target flow. If so, it
is likely that the target flow traverses the nodes at which
congestion/fluctuation has been caused.

\subsection{Chaff-based Defenses}
Adding chaff traffic (also referred to as padding traffic or dummy traffic) is a defense mechanism that thwarts traffic analysis 
by concealing real traffic patterns. An important family of chaff-based 
anonymity protocols uses \emph{link padding}~\cite{shmatikov2006timing, wang2008dependent, freedman2002tarzan,aqua,herd2015}.
Link padding, used together with link encryption, allows neighboring forwarding nodes to add chaff to shape
the patterns of all traffic on a network link into either constant-rate 
transmission~\cite{shmatikov2006timing, wang2008dependent} or a predetermined
packet schedule~\cite{freedman2002tarzan,aqua,herd2015}. However, because in link padding a node is able to
distinguish chaff packets from real packets, attackers that compromise nodes are still capable of 
identifying the underlying traffic patterns and conduct traffic analysis.

Another class of chaff-based protocols uses \emph{end-to-end padding}~\cite{levine2004timing}.
In the end-to-end padding scheme, end hosts craft chaff packets that traverse the network
together with real packets, and the added chaff packets carry flags to inform the forwarding nodes
about when to drop the chaff packets. Thus, an end host's traffic demonstrates different
patterns as the traffic traverse the network. Compared to link padding, in end-to-end padding
a compromised node cannot distinguish chaff traffic from real traffic, and
is thus unable to discover the real traffic patterns. Nevertheless, the existing work,
defensive dropping~\cite{levine2004timing}, fails to fully conceal the timing information
of the real traffic, and is trivially defeated by measuring packet density~\cite{shmatikov2006timing}.

\section{Problem Definition}
\label{sec:problemdefinition}

We consider a scenario where an adversary secretly conducts a 
network mass-surveillance program. By stealthily tapping into inter-continental fiber links, or
by controlling
a set of domestic ISPs/IXPs, the adversary gains bulk access to network traffic. 
Besides matching identifiers to filter packets, the adversary is also capable of 
conducting traffic manipulation and traffic pattern matching. 
A pair of anonymity-conscious users would like to communicate through the network,
hiding the fact that they are communicating from the adversary. The communication between the pair of users is bi-directional. Without
loss of generality, we call the user that initiates the anonymous communication \emph{sender},
and the other user \emph{receiver}.

\subsection{Network Assumptions}
The underlying network is divided into ASes, or simply \emph{nodes}. Each node forwards
packets according to a routing segment. Each routing 
segment contains forwarding information for a node between the sender and the receiver.
For a sender to reach a receiver,
the sender can obtain a sequence of \emph{routing segments}, named \emph{path}. 

Except the ingress and egress links that are needed as forwarding information
through an AS, routing segments should leak no extra information about the end
hosts or the path before or after the forwarding node. This property is satisfied by
several next-generation Internet architectures that use source-controlled
routing (e.g., SCION~\cite{Xin2011SCION}, NIRA~\cite{Yang2007NIRA}, or
Pathlet~\cite{godfrey2009pathlet}), or in the Internet through IPv6 Segment
Routing~\cite{segment_routing}.

\subsection{Threat Model}

We consider a global active
adversary, that is capable of
controlling all links between any pair of ASes, or between an AS and an end host.
This means that the adversary has bulk access to contents and timing information
of packets on all links and can also inject, drop, delay, replay, and modify
packets. We additionally assume that the adversary is able to
compromise a fraction of ASes. By compromising an AS, the adversary learns all
keys and settings, has access to all traffic that traverses the compromised AS,
and is able to control the AS including delaying, redirecting, and dropping traffic,
as well as fabricating, replaying, and modifying packets. We only guarantee
relationship anonymity for end hosts if there exists at least one uncompromised
AS on the path between sender and receiver. We remark that the adversary
under this assumption is able to perform all traffic analysis attacks in Section~\ref{sec:ta_techniques}.

\subsection{\name Goals} \label{sec:goals}

\paragraph{Anonymity} \name aims to provide relationship anonymity 
(defined by Pfizmann and K\"{o}hntopp~\cite{pfitzmann2001}) when
a sender and a receiver share mutual trust. 
We refer to the relationship anonymity under this condition as
\emph{third-party relationship anonymity}. While requiring trust in receivers
limits our protocol's application scope,
third-party anonymity is actually sufficient when communicating parties are
authenticated end-to-end (e.g., VoIP), when avoiding censorship where the
receiver (e.g., a foreign news site)
is known not to cooperate with the censoring entity, when a warrant canary (e.g.,
\url{www.rsync.net/resources/notices/canary.txt}) has been recently updated for that endpoint, or when 
the receiver is a trusted node acting as a proxy.

\paragraph{High throughput and low latency}
The processing overhead should be small, i.e., it should only require symmetric
cryptographic operations and access to a small amount of easy-to-manage per-flow
state. Consequently, an efficient implementation (running at line speed) on a
network device should be possible with a small amount of extra hardware.

\paragraph{Scalability} Nodes should be capable of handling the large
volume of simultaneous connections as observed on Internet core routers.
\name aims to minimize the amount of per-flow state maintained. 
Specifically, \name guarantees that the amount of state on a router is bounded given a fixed throughput.
Moreover, adding new nodes to the network should additionally not require coordination with all other nodes.

\section{Protocol Design}
\label{sec:overview}

\paragraph{Communication Model}
Hosts communicate anonymously through \name{}-enabled
Autonomous Systems (ASes) using \emph{flowlets}. A \name flowlet allows an end host to
send traffic anonymously at a constant rate $B$ for a fixed time period $T$. All 
anonymous traffic is divided into a set of flowlets by end hosts to leverage \name's service.
Figure~\ref{fig:overview} graphs the lifecycle of a \name flowlet. 

A flowlet's life-cycle begins with a \emph{setup phase} followed by a \emph{data transmission phase}. At the beginning of the setup phase, a sender first anonymously retrieves two paths: a forward path from the sender to the receiver and a backward path from the receiver back to the sender. A path contains the routing segments, the public keys, and the certificates of all nodes between the two end hosts. One mechanism for anonymously retrieving
paths is to have end hosts query global topology servers through \name
flowlets that are established using network configuration information (e.g.,
distributed to end hosts through a DHCP-like
infrastructure~\cite{Chen2015Hornet}). Another mechanism is to disseminate paths
and public keys throughout the network to end hosts, as done in certain future
network architectures (e.g., NIRA~\cite{Yang2007NIRA}, Pathlets~\cite{godfrey2009pathlet}).
A third mechanism could be based on private
information retrieval (PIR)~\cite{Chor1998},
which allows to trade off a lower communication
overhead for an increased computation overhead on the servers providing the network information
and the keys.

Once the sender successfully obtains both paths, the sender and the receiver exchange two setup messages traversing the obtained paths. By processing a setup message, each on-path node establishes a shared symmetric key with the sender. The per-node shared key is later used to conceal routing information by layered encryption/decryption in the data transmission phase. To prevent storing per-flow cryptographic state on each node, a node encrypts the shared key using a local secret key that the node never reveals. The resulting encrypted shared key, which we call the \emph{Forwarding Segment} (FS),  is carried by all data packets and allows the node to dynamically retrieve its shared symmetric key.

\begin{figure}[tbp]
	\centering
	\includegraphics[width=0.45\textwidth]{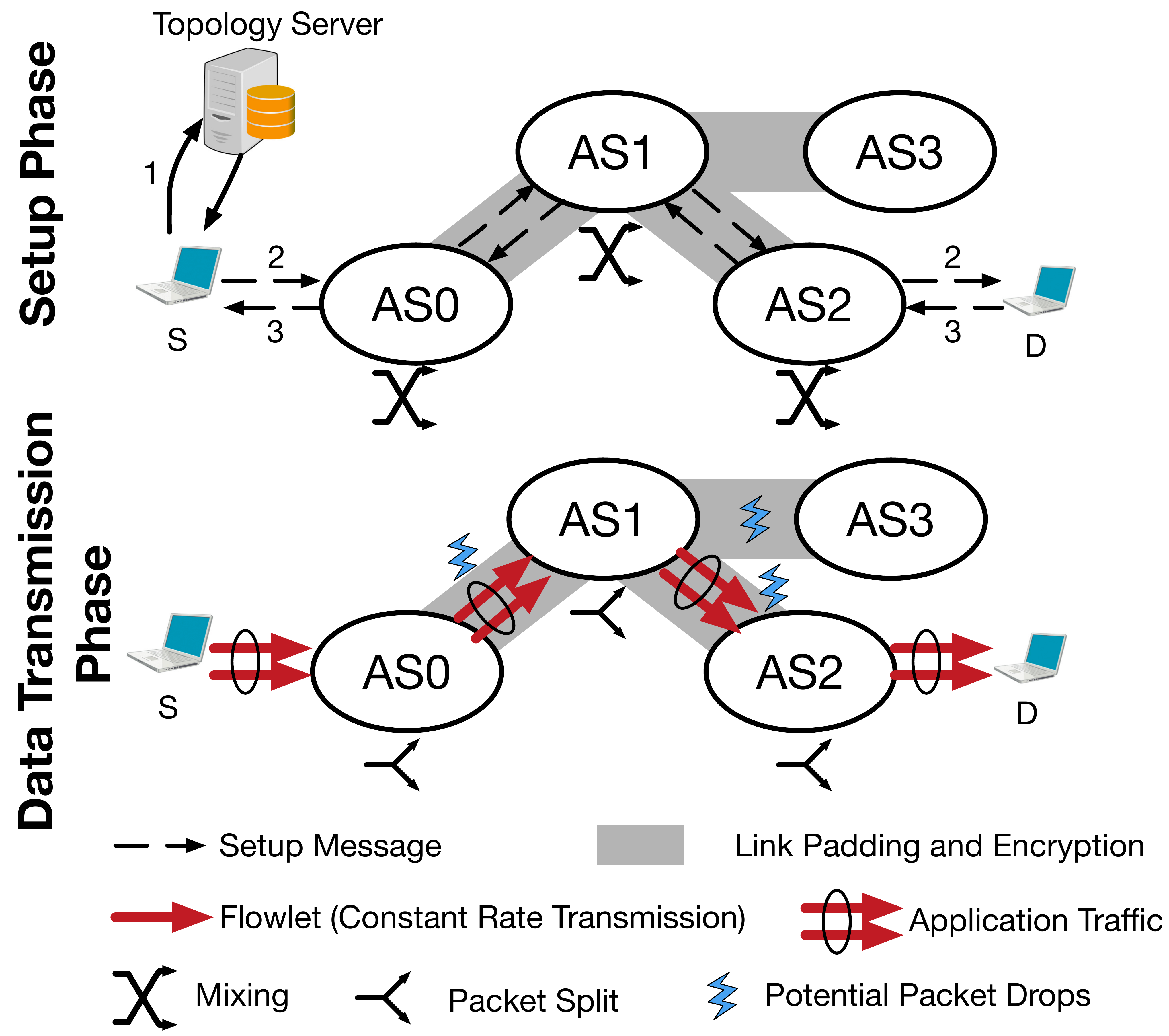}
	\caption{\name design overview.}
	\label{fig:overview}
\end{figure}

With routing segments, FSes, and per-node symmetric keys, the sender is able to create \name data packets that can reach the receiver. 
An on-path node can process a data packet with only symmetric cryptographic operations, enabling highly efficient packet forwarding.
Within the first batch of packets along the forward path, the sender transmits all routing segments, FSes, and shared symmetric keys for the 
backward path, so that the receiver can send packets back to the sender.

\paragraph{Traffic analysis resistance}
\name resists traffic analysis attacks by combining
an onion routing protocol (an enhanced adaptation of the one in \hornet), 
a newly proposed end-to-end traffic shaping scheme, and mixing. 
First, compared to \hornet which provisions confidentiality, authenticity, and bit pattern unlinkability,
\name additionally offers payload integrity protection, replay protection, and packet splitting,
which is a vital enabling technique for the end-to-end traffic shaping scheme (Section~\ref{sec:onion_routing}). 
Second, for the data transmission phase, \name enables end-to-end traffic normalization for flowlet traffic.
For each flowlet, the sender and receiver maintain a constant transmission rate
shared by every end host. Each forwarding node
maintains the same constant transmission rate for outgoing packets belonging to the flowlet (Section~\ref{sec:flowlet}).
Third, for messages in the setup phase, 
\name requires each node to conduct mixing~\cite{chaum1981untraceable} in order to prevent linking messages based
on their timing and order (Section~\ref{sec:mixing}). Finally, to hide the difference between setup packets and data packets and to defeat a global eavesdropper that monitors the number of flowlets on links between nodes, \name additionally requires neighboring nodes to perform link encryption and link padding (Section~\ref{sec:link_encryption}).

The rationale for adopting different techniques for the setup phase and the data transmission phase 
is due to our observation of the different performance requirements in these two phases. 
Regarding the setup phase, assuming a large number of simultaneous connection 
setups, batching setup messages on a node will result in a small delay for the setup phase.
Moreover, because changing the order of messages received by a node has no impact
on the performance of the setup phase, we can randomize the order of messages within
each batch. Finally, since processing a chaff setup message requires public-key cryptographic
operations, creating chaff setup messages would result in a large computational overhead.

For the data transmission phase, on the other hand, because packet order is important
for TCP performance, randomizing the message order severely impacts 
application performance. Additionally, because data
packet processing is highly efficient, we can actively conduct traffic shaping on
both end hosts and intermediate nodes by using chaff packets
(Section~\ref{sec:flowlet}).

\subsection{\name Onion Routing Protocol}
\label{sec:onion_routing}
Like the \hornet onion routing
protocol~\cite{Chen2015Hornet}, the \name protocol offers bit-pattern
unlinkability, payload confidentiality, and per-hop authenticity. Bit-pattern
unlinkability eliminates any identifiers that facilitate packet matching.
Payload confidentiality prevents leaking upper-layer sensitive information.
Finally, each \name header contains per-hop MACs that protect the integrity of
both the header and the payload, unlike \hornet, whose per-hop integrity
guarantees only cover the header. Therefore, in \name, tampered or
forged packets will be detected by benign nodes on the path and dropped
immediately.

\name also adopts the scalable
design of \hornet, i.e., using packet-carried forwarding state. Storing
per-flowlet state at core routers requires a large amount of high-speed memory,
precluding scalability. Thus, in line with state-of-the-art
network-layer anonymity protocols~\cite{Hsiao2012, Sankey2014, Chen2015Hornet},
\name embeds all necessary forwarding state (e.g., onion decryption keys,
next-hop information, control flags) in packet headers instead of storing the state
on routers.

\begin{table}[!!tbh]
	\centering
	\scriptsize
	\begin{tabular}{c|cccccc}
		\hline
		\thead{Protocol} & \makecell{Bit-pattern \\ unlinkability} & Scalability & \makecell{Payload\\ Integrity}  & \makecell{Replay \\Protection} & \makecell{Packet \\ Splitting} \\
		\hline
		\hornet & Yes & Yes & No & No & No \\
		\name & Yes & Yes & Yes & Yes & Yes \\\hline
	\end{tabular}
	\caption{Comparison between \name and \hornet onion routing protocols}
\end{table}

We highlight three new features that \name introduces for the data transmission
phase compared to \hornet. First, integrity protection
is extended to data packets' payloads, eliminating tagging attacks targeting at
manipulating data payloads to create recognizable patterns.
Second, data packets within the same flowlet have unique identifiers bound to
the packets themselves, enabling replay protection.
Third, \name allows an end host to create special chaff packets, each of which splits into two
packets at a specific node. To all other nodes, the
original packet and the resulting packets are indistinguishable from
ordinary data packets in the same flowlet.
Split packets traverse the same path as other packets in the flowlet and their per-hop MACs need to be correct at each downstream node. Splitting a chaff
packet into multiple packets plays a vital role in the end-to-end traffic
shaping technique described in Section~\ref{sec:flowlet}. We defer the detailed
description of the technical aspects of packet splitting to
Section~\ref{sec:specification}.

\paragraph{Replay protection}
In \name, each \name packet header is uniquely
identifiable, enabling intermediate nodes to detect replay attacks by checking
the header's freshness. Specifically, an intermediate node can retrieve 3 fields
from each packet: (1) a shared secret with the sender, (2) a per-packet Initial
Vector (IV), and (3) a per-packet expiration time. The first two fields together
uniquely identify a packet and are used as input to membership queries and
for the insertions to the replay detector. The third field is used to check and drop
expired packets.

\name nodes detect replayed packets by maintaining a rotating Bloom
filters composed of 3 subject Bloom filters, as described by Lee et al.~\cite{lee2017replay}.
A packet received at $t=[i \cdot \frac{TTL}{2},\, (i+1) \cdot \frac{TTL}{2}]$ is checked
against all 3 filters and is only inserted into $i'$-th Bloom filter, where $i' \cong i \mod 3$. 
The $i$-th subject filter is cleared at time $(3N + i) \cdot \frac{TTL}{2}$ (N is an integer).
The rotating Bloom filter guarantees that each packet inserted has a lifetime between $TTL$ and $\frac{3}{2}TTL$,
where $TTL$ is the maximum lifetime of a packet. 
To reduce cache misses and increase performance, we also use \emph{blocked Bloom
filters}~\cite{Putze2007} instead of standard Bloom filters.

Replay detection state is not per-flow state, since the size of the detector
grows linearly with its node's bandwidth, and not with the number of flowlets
traversing that node. The size of our detector is
\texttildelow15~MB\footnote{Computed using the CAIDA dataset described in
Section~\ref{sec:evaluation}.} for a 10~Gbps link when the
false positive rate is at most $10^{-6}$ and $TTL=6~\text{s}$ (the maximum
packet lifetime we consider in Section~\ref{sec:exptime}). Each false positive
result causes the corresponding packet to be dropped. Given that the packet drop rate
of the Internet is around 0.2\%~\cite{Sundaresan2011}, we could reduce the 
detector's size by allowing higher false positive rate.

\subsection{End-to-end Traffic Shaping}
\label{sec:flowlet}

\paragraph{Flowlet}
Our basic idea for defending data transmission against traffic analysis is to shape
traffic from heterogeneous applications into constant-rate transmission. 
A flowlet is the basic unit
through which an end host is able to transmit
packets at a constant throughput $B$ and for a maximum lifetime $T$. During the
lifetime $T$ of a flowlet, the end host always transfers packets at rate
$B$, inserting chaff packets if necessary. More
generally, if an end host needs to transfer data at rate $B'$ for time $T'$, 
it initiates a sequence of $\lceil T'/T \rceil$ flowlet batches, each of which contains $\lceil
B'/B \rceil$ simultaneous flowlets.

An end host shuts down a flowlet before the flowlet expires
when there is no more data to send. When shutting down multiple simultaneous
flowlets, an end host pads each flowlet with a random number of packets to
prevent linking the flowlets by their expiration times.
A node erases local state and terminates a flowlet
when there are no more packets in its outgoing packet queue.

The key property of a flowlet is to maintain constant transmission rates not only
at end hosts but also on all traversed links, for which the flowlet 
relies on end-to-end padding instead of link padding.
In link padding, a pair of neighboring intermediate nodes coordinate
to inject chaff to maintain a constant sending rate on a link.
While link padding is effective against a network adversary,
it is insufficient in the case of compromised nodes, since
they can distinguish chaff inserted by neighbors from actual data packets. 
To defend against compromised nodes, we need chaff packets that are indistinguishable from data packets. Because \name uses onion encryption as a basic building block, 
one can create such indistinguishable chaff only when possessing shared keys with all traversing nodes. Thus, only sending end hosts are able to create such chaff.

\paragraph{Necessity of packet splitting}
To achieve constant-rate transmission, every flowlet should ideally arrive and leave with rate $B$ at every node.
However, drops/jitter may cause the incoming rate to vary: a higher rate is
absorbed by the queues, but a lower rate requires that the node be able to
produce ``extra packets'', which need to resemble legitimate packets to any
downstream node. This implies that these packets must also be generated by the
sender like end-to-end chaff.
 But since the sender cannot send at a rate higher than $B$, it cannot
send additional packets for the nodes to cache and use when needed. The only
option then seems to be to have very long queues, and let each node fill a
significant fraction of them with packets when the transmission of the flowlet
first begins, before the node starts forwarding packets for that flowlet.
However, this requires far too much state, and also adds significant latency in
terms of time to the first byte,  making this option unfeasible. The
apparent dilemma can be solved with a technique we call \emph{packet
	splitting}. 

The packet splitting technique allows an end host to create a packet
that can be split into two packets at a specific intermediate node.\footnote{The general
packet splitting technique supports a n-way split. We consider only two-way packet splits
because of limited Maximum Transmission Units (MTU) in the network.}
The resulting packets should be indistinguishable from other non-splittable packets.
This requirement indicates that the resulting packets should still traverse 
the same path and reach the recipient's end host. We present the
algorithm to split packets in Section~\ref{sec:specification}.

\paragraph{Traffic shaping for flowlet outgoing rate}
To enable end-to-end traffic shaping, for each on-path node $n_i$, an end host
selects a slot in its transmission buffer with probability $Pr_i^{split}$
and fills in a newly generated splittable chaff packet that will split at node
$n_i$. As an optimization, the end host can also select a slot that already
contains chaff packets and replace it with splittable chaff packets. When a node
receives a packet that should be split at the node, the node performs
the split and caches resulting packets in its chaff packet queue.

Each node maintains a per-flowlet chaff queue of cached chaff packets. To
guarantee an invariant outgoing flowlet rate, nodes periodically output a data
packet from the data packet queue. In case that the data packet queue is empty,
the node outputs a chaff packet from the flowlet's chaff queue. We limit the chaff queue size by a
maximal length  $L_{\mathit{chf}}$.
In the (unlikely) scenario where the chaff queue is also empty, a local
per-flowlet failure counter $h$ is increased. When $h$
exceeds a threshold $H$ negotiated during flowlet setup, the node terminates the
flowlet. $H$ is a security parameter of the flowlet that determines how
sensitive the flowlet is against potential malicious packet drops.

When a node shuts down a flowlet, an intermediate node no longer receives packets from
upstream nodes. It will first drain its local chaff packet queue and then terminate
the flowlet when the threshold $H$ is reached. We remark that such a termination process
results in successive termination on nodes and small variable intervals between
termination times on different nodes because of the variable number of cached chaff packets.

We remark that both the chaff queues and failure counters constitute per-flow state.
Nevertheless, the amount of state stored on a node is bounded by the node's bandwidth.
Because each flowlet consumes a fixed amount of bandwidth, a node
with fixed total bandwidth is only capable of serving a fixed number of flowlets. 
Thus, the amount of state that a node maintains for its flowlets is bounded given
its total available bandwidth. Accordingly, a node will have to refuse setup messages
for new flowlets if its bandwidth is already fully occupied.
We evaluate the amount of state the queues require in detail in Section~\ref{sec:evaluation}.

\subsection{Mixing in the Setup Phase}
\label{sec:mixing}

Each \name node applies a basic form of mixing when processing setup messages.
After a setup message is processed by an intermediate node, the
node queues the message locally into batches of size $m$. Once
there are enough setup messages to form a batch, the node first 
randomizes the message order within each batch and then sends out the batch. 

Through batching and order randomization, a \name node aims to obscure the timing 
and order for setup messages. An adversary that observes both input and output 
setup messages of a non-compromised node cannot match an output packet to 
its corresponding input packet within the batch.

The batching technique introduces additional latency because the setup messages
have to wait until enough messages are accumulated. Assume that $r_{setup}$
is the number of incoming setup messages every second, the added latency
can be computed as $\frac{m}{r_{setup}}$. Given the large number of simultaneous
connections within the network, the introduced latency is very low,
as shown by our evaluation in Section~\ref{sec:setupevaluation}.

\subsection{Link Encryption and Padding}
\label{sec:link_encryption}
Each pair of neighboring \name nodes agree upon a constant transmission rate upon link
setup. The negotiated transmission rate determines the maximum total rate for data packets.
When the actual transmission rate exceeds the negotiated rate on a link, the sending node drops
the excessive packets. When the actual transmission rate is lower than the negotiated
rate, the sending node will add chaff traffic.
The chaff traffic inserted by an intermediate node to shape traffic on a link only traverses
the link and is dropped by the neighboring node. 

To prevent an adversary observing a link between
two honest nodes from distinguishing chaff traffic
from actual data traffic, all pairs of neighboring nodes negotiate a symmetric key through
the Diffie-Hellman protocol, and use it to
encrypt all packets transmitted on their shared link. This also makes setup messages
and data packets indistinguishable. 

As an optimization to reduce chaff traffic and improve bandwidth usage, we
additionally allow neighboring nodes to agree on a schedule of transmission
rates as long as transmission rate is detached from the dynamics of individual
traffic rates. For example, because the actual link rate on a link often
demonstrates similarity at the same time of different days, we can reduce the
amount of chaff traffic by setting the transmission rate between $[t, t']$ to $ \mathcal{B}_{[t, t']} + k\cdot \Sigma_{[t, t']}$.
$\mathcal{B}_{[t, t']}$ is the historic average transmission rate between $[t, t']$,
$\Sigma_{[t, t']}$ is the standard deviation for the transmission rate,
$k$ is a factor that allows administrators to account for temporal changes
of the bandwidth usage.

\section{Protocol Details}
\label{sec:specification}
This section presents the details of \name data packet formats and processing functions. We show how to create a fixed-size packet that can be split into two new packets of the same size whose per-hop MAC can still be verified. Using the packet processing functions, we present the \name data transmission phase on end hosts and intermediate nodes.

\subsection{Notation}

\newcommand{\MAC}{\ensuremath{\mathsf{MAC}}}
\newcommand{\PRG}{\ensuremath{\mathsf{PRG}}}
\newcommand{\PRP}{\ensuremath{\mathsf{PRP}}}
\newcommand{\ENC}{\ensuremath{\mathsf{ENC}}}
\newcommand{\DEC}{\ensuremath{\mathsf{DEC}}}
\newcommand{\randomgen}{{\sc rand}}
\newcommand{\expTime}{\mbox{\sc exp}}
\newcommand{\ExpTime}{$\expTime$\xspace}

\providecommand{\str}{\ensuremath{\sigma}}

We first describe our notation. In general, $sym^{dir}$
stands for the symbol $sym$ of a specific direction $dir\in\{f,b\}$, which is
either forward (src to dst) or backward (dst to src).
$sym_{i}^{dir}$ indicates the symbol $sym$ belongs to $i$-th node $n_i^{dir}$
on the path in direction $dir$. For simplicity, we denote the set of all $sym$ for a path
$p^{dir}$ as $\{sym_i^{dir}\}$.  We also define a series of string operations:
$0^{z}$ is a string of zeros with length $z$; $|\str|$ is the length of the
string $\str$; $\str_{[m..n]}$ refers to the substring between $m$-th bit to
$n$-th bit of string $\str$ where $m$ starts from 0; $\str_1\concat \str_2$ stands for concatenation
of string $\str_1$ and $\str_2$. Table~\ref{tab:notation} summarizes the notation in
this paper.

\begin{table}[htbp]
	\scriptsize
	\begin{tabular}{c|m{6cm}}
		\hline
		Symbol & Meaning \\\hline
		$k$ & security parameter used in the protocol \\
		$c$ & size of per-hop segment \\
		$b$ & size of control bits and expiration time \\
		$r$ & maximum path length permitted by the protocol\\
		$m$ & fixed-size of a data packet payload  \\
		$p^{dir}$ & path of a specific direction $dir$ \\
		$l^{dir}$ & length of a path $p^{dir}$ \\
		$n_i^{dir}$ & the $i$-th node on path $p^{dir}$ \\
		$x_i^{dir}$, $g^{x_i^{dir}}$ & the private and public key pair of node $n_i^{dir}$\\
												$h_{op}$ & a hash function to generate the key for $op$\\
				$R$ & routing segment, e.g., the ingress and egress ports\\
		
				$\expTime_{i}$ & expiration time for a packet at node $n_i$\\
		$FS$ & forwarding segment \\
		$s$ & a symmetric onion key shared with the sender\\
		$IV$ & per-packet initial vector \\
		$\gamma$ & per-hop \MAC\\
		$\beta$ & the opaque component of a packet header \\
		$O$ & onion data packet \\
		\hline
	\end{tabular}
			\caption{Notation used in the paper.
					}
	\label{tab:notation}
\end{table}

\subsection{Initialization \& Setup Phase}

In the setup phase, the sender node aims to anonymously establish a set of shared keys
$\{s_i^{dir}\}$ with all nodes on the forward and backward path, and a shared
key $s_{SD}$ with the receiver. In the following protocol description and in our
implementation, we use \hornet{}'s Sphinx-based single-round-trip
setup~\cite{Chen2015Hornet}. Note that we can also set up flowlets using Tor's telescopic method~\cite{dms04} which increases latency, but preserves perfect forward secrecy.

Once the setup phase is complete, in addition to the shared keys, the sender also obtains from each node on
both paths a \emph{Forwarding Segment} (FS)~\cite[Section~4]{Chen2015Hornet}.
The FS created by the node $n_i^{dir}$ contains the key shared between the
sender and that node $s_{i}^{dir}$ and the routing information $R_i^{dir}$ which tells the
node how to reach the next hop on the path. The FS is encrypted using a secret
value known only to the router that created the FS.
As shown in Section~\ref{sec:datacommunicationphase}, these FSes are included in every data packet: each node can then the retrieve the FS it
created, decrypt it, and recover the packet processing information within.
Unlike \hornet, we do not store the expiration time \ExpTime in a FS,
but include it alongside the FS in the packet (see Section~\ref{sec:datapacketformat}).
This allows the sender to set a different expiration time for each packet and limit
the time window in which the packet is valid, which is necessary for replay
protection.

\subsection{Data Packet Processing}
\label{sec:sub:data_packet}
\subsubsection{Requirements}
\name data packets are fixed-size onion packets whose integrity is protected by per-hop MAC. Processing these packets should satisfy the following three requirements:
\begin{itemize}
	\item An output packet cannot be linked to the corresponding input packet without compromising the processing node's local secret value.
	\item Processing a packet cannot leak a node's position on the path.
	\item Processing a packet cannot change the packet size regardless of underlying operations.
\end{itemize}

The last requirement is particularly challenging to satisfy, since \name allows flow mutations. Consider the split operation, which takes a fixed-size packet and creates two uncorrelated packets of the same size. The splitting procedure needs to ensure that subsequent nodes can verify the MACs in both new packets.

\subsubsection{Data packet format}\label{sec:datapacketformat}
\name data packets are shown in Figure~\ref{fig:pkt_format}.
At the beginning of each packet is an $IV$ field that carries a fresh initial
vector for each packet in a flowlet. After the $IV$ field are four fields
that form an onion layer: an FS, a per-hop MAC, control bits, and the expiration time. The rest of
the fields, including the rest of header information, padding bits, and the
payload, are encrypted, and are thus opaque to the processing node.

When a packet arrives, the first three fields are accessible to a node without requiring cryptographic processing, so we call
these fields as public state. The control bits and the expiration time are only available after the node decrypts the packet, so they are called secret state. In addition, each
header is padded to a fixed size regardless of the actual number of nodes on
the path, and the padding bits are inserted between the header and the payload.

\begin{figure}
	\centering
	\includegraphics[width=0.40\textwidth]{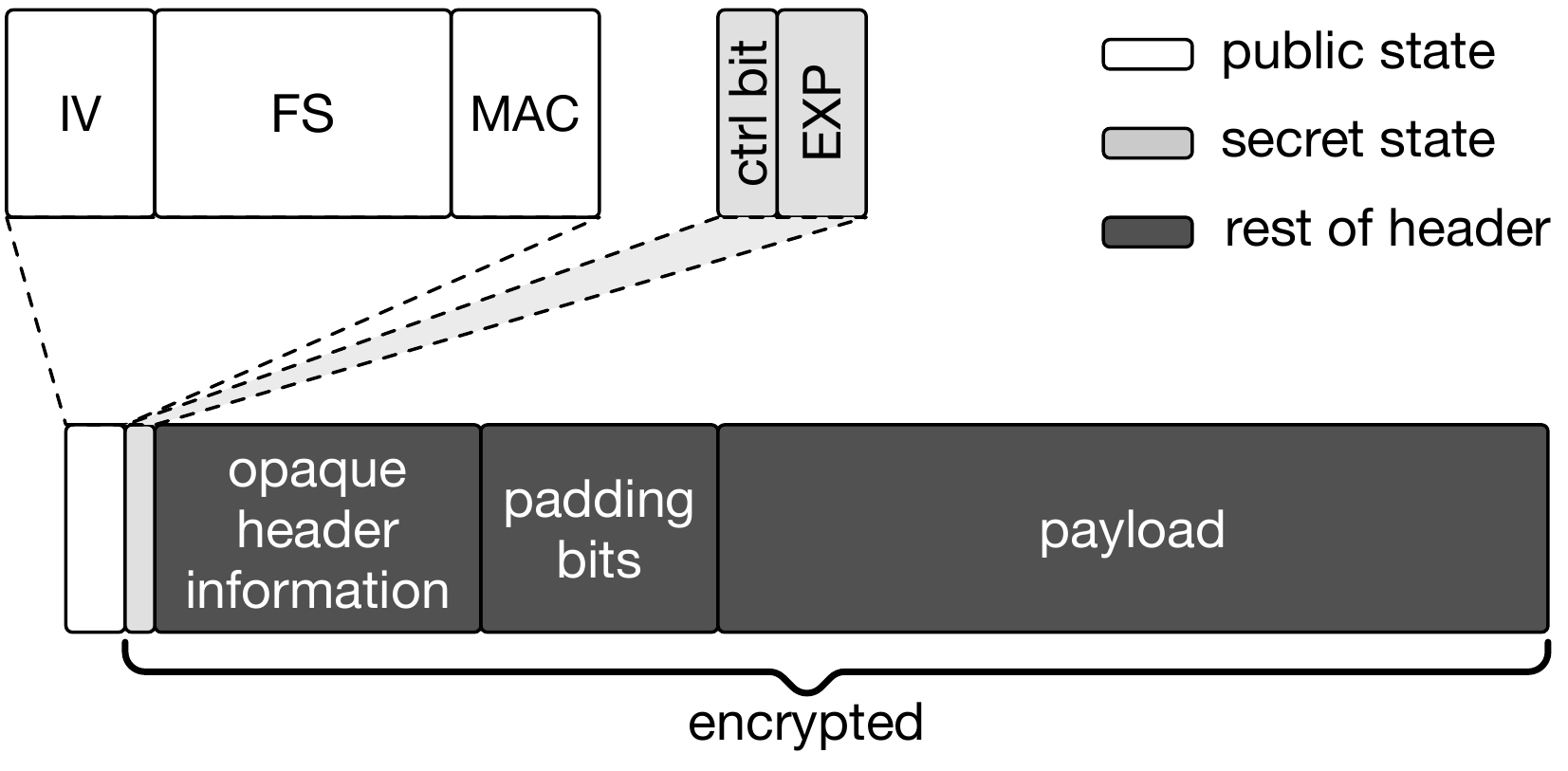}
	\caption{\name packet format.}
	\label{fig:pkt_format}
\end{figure}

\newcommand{\removeOnionLayer}{\mbox{{\sc remove\_onion\_layer}}}
\newcommand{\RemoveOnionLayer}{$\removeOnionLayer$\xspace}

\subsubsection{\name packet creation}
Both end hosts generate data packets using a subroutine shown in
Algorithm~\ref{alg:create_pkt}. The subroutine creates an onion packet to be
forwarded from node $n_k$ to node $n_l$. For each onion layer, it computes a per-hop
MAC (Line~\ref{alg:create_pkt:gamma}) and onion-encrypts both the header
(Line~\ref{alg:create_pkt:beta}) and the payload
(Line~\ref{alg:create_pkt:encrypt}).

One important feature of this onion encryption algorithm is to add per-hop
state (specifically, an FS, a MAC, control bits and an expiration time) to the packet header
without changing its total size. The function achieves this feature by
strategically pre-computing the padding bits in the header
(Line~\ref{alg:create_pkt:padding}) to ensure that the trailing $c$ bits of
header after encryption are always equal to $0^c$. As a result, the trailing
zero bits can be truncated without losing information when the header is
encrypted again (Line~\ref{alg:create_pkt:beta}).

\newcommand{\createDataPacketRoutine}{\mbox{\sc create\_onion\_routine}}
\newcommand{\CreateDataPacketRoutine}{$\createDataPacketRoutine$\xspace}

\begin{algorithm}\footnotesize
	\begin{algorithmic}[1]
		\Procedure{create\_onion\_routine}{}
		\Statex Input: $\{s_i\}$, $\{FS_i\}$, $\{ctrl_i\}$, $\{\expTime_i\}$, $IV$, $k$, $l$, $O$
		\Comment{with $k<l$}
		\Statex Output: $(IV_k, \FS_k, \gamma_k, \beta_k, O_k)$
		\State $\phi_k \gets \varepsilon$
		\State $IV_k \gets IV$
		\For{$i \gets k+1,\dots, l$}
		\State $IV_i \gets \PRP(h_{\PRP}(s); IV_{i-1})$
		\State $\phi_i \gets (\phi_{i-1} \concat 0^{c}) \;\oplus$
		\StatexIndent[3] $ \PRG(h_{\PRG}(s_{i-1}\oplus IV_{i-1}))_{[(r-i-1)c + b .. rc + b - 1]}$ \label{alg:create_pkt:padding}
		\EndFor

		\State $\beta_{l} \gets \big\{$\randomgen$(c(r-l-1)) \concat \phi_{l}\big\}$ \label{alg:create_pkt:start}
		\State $O_{l} \gets \ENC(h_{\ENC}(s_l); IV_{l}; O)$
		\State $\gamma_{l} \gets \MAC(h_{\MAC}(s_{l} \concat IV_{l}); \FS_l \concat \beta_{l} \concat O_{l})$
		\For{$i \gets (l-1), \ldots, k$}
		\State $\beta_{i} \gets \left\{ctrl_i \concat \expTime_i \concat \FS_{i+1} \concat \gamma_{i+1} \concat {\beta_{i+1}}_{[0 .. c(r-2)-1]}\right\}$
		\StatexIndent[4] $\vphantom{0} \oplus \PRG(h_{\PRG}(s_i \concat IV_{i}))_{[0 .. b+(r-1)c-1]} $ \label{alg:create_pkt:beta}
		\State $\gamma_{i} \gets \MAC(h_{\MAC}(s_{i} \concat IV_{i}); \FS_i \concat \beta_{i} \concat O_{i})$
		\State $O_i \gets \ENC(h_{\ENC}(s_i); IV_{i}; O_{i+1})$
		\label{alg:create_pkt:encrypt}
		 \label{alg:create_pkt:gamma}
		\EndFor
		\EndProcedure
	\end{algorithmic}
	\caption{Create a partial data packet.}
	\label{alg:create_pkt}
\end{algorithm}

Normally, an end host creates a packet that traverses the whole path from the
first node $n_0^{dir}$ to the last node $n_{l^{dir}-1}^{dir}$. It generates
such a packet by setting $k=0$, $l=l^{dir}-1$, and all $ctrl_i =
\mbox{FWD}$ in function \CreateDataPacketRoutine.

\paragraph{Generate splittable packets}
Creating a data packet that can be split into two packet requires an end host
to first create two children packets and then merge them into a single packet.
Because we require all packets to have the same size, i.e., both children
packets have to be of the same size as their parent, the key challenge is to
guarantee that the per-hop MACs in the children packets successfully verify
even after the splitting node adds padding bits to the children packets. For
this reason, the splitting node generates padding bits by a PRG keyed by the
key shared with the end host, so that the end host can predict the padding bits
and pre-compute the per-hop MACs in both resulting packets accordingly.

\newcommand{\createSplittableDataPacket}{\mbox{\sc create\_splittable\_data\_packet}}
\newcommand{\CreateSplittableDataPacket}{$\createSplittableDataPacket$\xspace}

Algorithm~\ref{alg:create_split_pkt} shows the function to
create a splittable data packet. At a high level, \CreateSplittableDataPacket
invokes the \\
\CreateDataPacketRoutine three times: it first creates two 
children packets using
\CreateDataPacketRoutine
(Line~\ref{alg:create_split_pkt:create_left} and
\ref{alg:create_split_pkt:create_right}), merges the resulting packets into a
new payload (Line~\ref{alg:create_split_pkt:merge}), and finally executes
\CreateDataPacketRoutine again to generate the parent packet
(Line~\ref{alg:create_split_pkt:create_pkt}). To ensure the
correctness of the per-hop MACs in the children packets after the payloads are
padded, the function generates the padding bits using a PRG keyed by the shared key
between the end host and the splitting node so that the latter can re-generate the padding bits accordingly
(Lines~\ref{alg:split_onion:padding_left} and
\ref{alg:split_onion:padding_right}). After the MACs are computed for the
children packets, the deterministic padding bits are truncated so that two
children packets can fit into the payload of their parent packet.

\begin{algorithm}
	\footnotesize
	\begin{algorithmic}[1]
		\Procedure{create\_splittable\_data\_packet}{}
		\Statex Input: $\{s_i\}$, $\{FS_i\}$, $\{ctrl_i\}$, $\{\expTime_i\}$, $IV$, $IV_0$, $IV_1$ $O_0$, $O_1$, $k$
		\Statex Output: $(IV_0, \FS_0, \gamma_0, \beta_0, O_0)$

		\State $O_0 \gets O_0 \concat \PRG(h_{\PRG}((s_k \oplus IV_0) \concat$  $\mbox{``left''}))_{[0..\frac{m}{2} + rc - 1]}$
		\label{alg:create_split_pkt:padding_left}
		\State $(IV_0', FS_0', \gamma_0', \beta_0', O_0') \gets $
		\StatexIndent[2] $\createDataPacketRoutine(\{s_i, \forall i \ge k\}, \{FS_i, \forall i \ge k\},$
		\StatexIndent[3] $\{\mbox{FWD}\}, \{\expTime_i, \forall i \ge k\}, IV_0, k, l_{dir}-1, O_0)$
		\label{alg:create_split_pkt:create_left}
		\State $O_1 \gets O_1 \concat \PRG(h_{\PRG}((s_k \concat IV_1) \concat$
			$\mbox{``right''}))_{[0..\frac{m}{2} + rc - 1]}$
		\label{alg:create_split_pkt:padding_right}
		\State $(IV_1', FS_1', \gamma_1', \beta_1', O_1') \gets$
		\StatexIndent[2] $\createDataPacketRoutine(\{s_i, \forall i \ge k\},$
		\StatexIndent[3] $\{FS_i, \forall i \ge k\}, \{\mbox{FWD}\}, IV_1, k, l_{dir}-1, O_1)$
		\label{alg:create_split_pkt:create_right}

		\State $O' \gets (IV_0', FS_0', \gamma_0', \beta_0', \{O_0'\}_{[0..\frac{m}{2} + rc - 1]}) \concat \vphantom{}$
		\StatexIndent[3] $(IV_1', FS_1', \gamma_1', \beta_1', \{O_1'\}_{[0..\frac{m}{2} + rc - 1]})$
		\label{alg:create_split_pkt:merge}

		\State $(IV_0, FS_0, \gamma_0, \beta_0, O_0) \gets$
		\StatexIndent[2] $\createDataPacketRoutine(\{s_i, \forall i < k\},$
		\StatexIndent[3] $\{FS_i, \forall i < k\}, \{\mbox{FWD, \dots, FWD, SPLIT}\},$
		\StatexIndent[3] $\{\expTime_i, \forall i < k\}, IV, 0, k-1, O')$
		\label{alg:create_split_pkt:create_pkt}
		\EndProcedure
	\end{algorithmic}
    \caption{Create a data packet that can be split into two new packets.}
	\label{alg:create_split_pkt}
\end{algorithm}

\subsubsection{Onion layer removal}

Nodes remove onion layers when processing data packets.
It essentially reverses a single step of \CreateDataPacketRoutine.
Algorithm~\ref{alg:remove_onion_layer} details this five step process. First,
the intermediate node retrieves the symmetric onion key $s$ shared with the
sender (Line \ref{alg:remove_onion_layer:key}); second, the node verifies a
per-hop MAC using a key derived from $s$ (Line~\ref{alg:remove_onion_layer:mac}); third,
the node ensures that the packet's size remains unchanged by adding padding
bits to the header and decrypting the resulting padded header with a stream
cipher; fourth, the control bits are extracted
(Line~\ref{alg:remove_onion_layer:ctrl}); finally, the payload is decrypted
(Line~\ref{alg:remove_onion_layer:payload}) and the next initialization vector
is obtained by applying a \PRP~keyed with $s$ to the current $IV$
(Line~\ref{alg:remove_onion_layer:iv}).

Note that the onion layer removal algorithm is different from a simple
decryption in two ways. First, the size of the packet remains the same after
processing, which prevents leaking information about the total number of hops
between the sender and receiver. Second, the processing only happens at the
head of the packet, which reveals no information about the processing node's
position on the path.

\newcommand{\fsCreate}{\mbox{{\sc fs\_create}}}
\newcommand{\FsCreate}{$\fsCreate$\xspace}
\newcommand{\fsOpen}{\mbox{{\sc fs\_open}}}
\newcommand{\FsOpen}{$\fsOpen$\xspace}

\begin{algorithm}
	\footnotesize
	\begin{algorithmic}[1]
		\Procedure{remove\_layer}{}
		\Statex Input: $P$, $SV$
		\Statex Output: $ctrl$, $P^{o}$, $R$, $\expTime$
		\State $\{IV \concat \FS \concat \gamma \concat \beta \concat O\} \gets P$ \label{alg:remove_onion_layer:fs}
		\State $s \concat R \gets \PRP^{-1}(SV, \FS)$  \label{alg:remove_onion_layer:key}
		\State \textbf{check} $\gamma = \MAC(h_{\MAC}(s \concat IV); \FS \concat \beta \concat O)$ \label{alg:remove_onion_layer:mac}
		\State $\zeta \gets \{\beta \concat 0^{c} \} \oplus \PRG(h_{\PRG}(s \concat IV))_{[0 ... (r-1)c+b-1]}$ \label{alg:remove_onion_layer:decrypt}
		\State ${ctrl \concat \expTime \concat \FS' \concat \gamma' \concat \beta'} \gets \zeta$ \label{alg:remove_onion_layer:ctrl}
		\State $O' \gets \DEC(h_{\DEC}(s); IV; O)$ \label{alg:remove_onion_layer:payload}
		\State $IV' \gets \PRP(h_{\PRP}(s); IV)$ \label{alg:remove_onion_layer:iv}
		\State $P^o \gets \{IV' \concat \FS' \concat \gamma' \concat \beta' \concat O'\}$
		\EndProcedure
	\end{algorithmic}
	\caption{Remove an onion layer.}
	\label{alg:remove_onion_layer}
\end{algorithm}

Depending on the value of control bits $ctrl$, the intermediate node performs one of the following two actions: FWD, or SPLIT. A node can split a data packet into two new packets by Algorithm~\ref{alg:split_onion_layer}.
First, the payload is split into two new packets (Line~\ref{alg:split_onion:split}).
Then the node pads both newly generated packets to the fixed size $m$ using pseudo-random bits obtained from a PRG keyed by $s$ (Line~\ref{alg:split_onion:padding_left} and \ref{alg:split_onion:padding_right}).

\newcommand{\splitOnionPacket}{\mbox{\sc split\_onion\_packet}}
\newcommand{\SplitOnionPacket}{$\splitOnionPacket$\xspace}

\begin{algorithm}
	\footnotesize
	\begin{algorithmic}[1]
		\Procedure{split\_onion\_packet}{}
		\Statex Input: $O$, $s$, $IV$
		\Statex Output: $P^{o}_0$, $P^{o}_1$
		\State $\{P_0'  \concat P_1'\}  \gets O$ \label{alg:split_onion:split}
										\State $P_0^o \gets P_0' \concat \PRG(h_{\PRG}((s \concat IV) \concat \mbox{``left''}))_{[0..\frac{m}{2} + rc - 1]}$
		\label{alg:split_onion:padding_left}
						\State $P_1^o \gets P_1' \concat\PRG(h_{\PRG}((s \concat IV) \concat \mbox{``right''}))_{[0..\frac{m}{2} + rc - 1]}$
		\label{alg:split_onion:padding_right}
		\EndProcedure
	\end{algorithmic}
	\caption{Split a data packet into two new packets.}
	\label{alg:split_onion_layer}
\end{algorithm}

\subsection{Data Transmission Phase}
\label{sec:datacommunicationphase}
\subsubsection{End host processing}
To send packets to receiver $D$, sender $S$ first makes sure that the
flowlet has not expired. Then $S$ chooses a value $\expTime_{min}$,
which has to be larger than its local time plus the end-to-end
forwarding delay plus the maximum global clock skew.
We expect that adding 1~s to the local time would be adequate for most circumstances.
However, $S$ cannot set the packet expiration time to be equal at every
hop, as otherwise this value could be used as common identifier (which violates
the bit-pattern unlinkability property.
Instead, $S$ chooses an offset
$\Delta_i \in [0, \Delta_{max}]$ uniformly at random, for each node $n_i^f$ on the path.
For every packet sent out,  $S$ determines $\expTime_{min}$ and
computes $\expTime_i = \expTime_{min} + \Delta_i$ for each node.
The value $\Delta$ needs to be chosen large enough to ensure that the interval
$[\expTime_{min}, \expTime_{min} + \Delta]$ overlaps with the intervals of a
large number of other concurrent flows. We expect that $\Delta \approx 5$~s would
be a safe choice.\label{sec:exptime}

After determining $\{\expTime_i\}$, $S$ also needs to decide which flow mutation actions
the packet will adopt. In case of packet splitting, $S$ also needs
to decide where to split the packet. For a packet that is forwarded to
the receiver without being split, we denote the payload to send is $O$.
For a packet that is split, we denote the payloads of the children packets
as $O_0$ and $O_1$. Let $k$ be the index of the node where the packet is split. Accordingly, $ctrl_i=\mbox{FWD}$, $\forall i \neq k$.
Third, $S$ uses $s_{SD}$ to encrypt the payload. This end-to-end encryption prevents the last hop node from obtaining information about the data payload.
$S$ also generates a unique nonce $IV$ for the packet. If the packet is splittable, $S$ generates another two unique nonces $IV_0$ and $IV_1$.
Fourth, if the packet will be split, $S$ creates the packet $P$ by
	\begin{multline}\small
	P=\createSplittableDataPacket(\{s_i^f\}, \{FS_i^f\}, \\
	  \{ctrl_i^f\}, \{\expTime_i^f\}, IV, IV_0, IV_1, O_0, O_1, k)
	\end{multline}
If the packet will only be forwarded to the receiver without a splitting action, $S$ creates the packet $P$ by
	\begin{multline}\small
	P = \createDataPacketRoutine(\{s_i^f\}, \{FS_i^f\}, \{ctrl_i^f\}, \\
	\{\expTime_i^f\}, IV, 0, l^f - 1, O)
	\end{multline}
Finally, $S$ forwards $P$ to the first hop node towards the receiver.

The process by which $D$ sends packets back to $S$ is similar to the above
procedure, but $D$ will use the forwarding segments and onion keys for the
backward path. However, right after $S$ finishes the setup phase, $D$ has not
yet obtained $g^{x_S}$, $\{s_i^b\}$, nor $\{FS_i^b\}$. In the \name data
transmission phase, the first packet that $S$ sends to $D$ includes $x_S$,
$\{s_i^b\}$ and $\{FS_i^b\}$ as the payload.

When an end host ($S$ or $D$) receives a data packet $P$, it can retrieve the data payload $O$ from the packet by $O = P_{[rc..rc + m-1]}$
The resulting $O$ can thus be decrypted by $s_{SD}$ to retrieve the plaintext payload.

\subsubsection{Intermediate node processing}
When a node receives a data packet $P=(IV, FS, \gamma, \beta, O)$, with the local secret $SV$, it first removes an onion layer by
	\begin{equation}\small
	 ctrl, P^{o}, R, \expTime = \removeOnionLayer(P, SV)
	\end{equation}
Note that the MAC must check in \RemoveOnionLayer for the process to move on. Otherwise, the node simply drops the packet.
Then, the node checks $t_{curr} < \expTime$ and ensures that the flowlet has not expired.
Afterwards, the node checks the control bits belonging to the current hop. If $ctrl=\mbox{SPLIT}$, the resulting payload $P^o$ must contain two sub packets.
The node creates two children packets $P_0^o$, $P_1^o$:
	\begin{equation}\small
	\{P_0^o, P_1^o\} = \splitOnionPacket(O, s, IV)
	\end{equation}
Lastly, if the packet is not dropped, the node forwards the resulting packet according to the routing decision $R$.
 
\section{Security Analysis}
\label{sec:security}
We discuss \name's defenses against passive
(Section~\ref{sec:sub:passive_attacks}) and active attacks
(Section~\ref{sec:sub:active_attacks}). 
We also conduct a quantitative analysis of 
\name's anonymity set size using the Internet 
topology and real-world packet traces (Section~\ref{sec:sub:ass}).
Our result shows that \name's anonymity set is 4 to 2$^{18}$ times
larger than those of \lap and \dovetail.
Finally, we present a formal proof that the \name protocol 
conforms to an ideal onion routing protocol defined by Camenisch
and Lysyanskaya~\cite{camenisch2005formal}.

\subsection{Defense against Passive Attacks}
\label{sec:sub:passive_attacks}

\paragraph{Flow dynamics matching}
In flow-dynamics matching attacks~\cite{danezistrafficanalysis2004,
	murdoch2007sampled}, adversarial nodes can collude to match two observed
flows by their dynamics, such as transmission rate. \name prevents such attacks
by normalizing the outgoing transmission rate of all flowlets through the use
of chaff traffic. Adversarial nodes are unable to distinguish chaff traffic from
real traffic. Accordingly, no flow dynamics are available to the adversary to
perform matching.

\paragraph{Template attacks}
\name enables end hosts to shape their traffic by adding chaff packets to hide
their real traffic patterns. The resulting traffic pattern of an outgoing
flowlet is uniform across the network. In addition, all \name packets have the
same length, preventing information leakage from packet length. The combination
of these two features completely neutralizes template attacks.

\paragraph{Network statistics correlation}
These attacks rely on the capability of the adversary to observe macroscopic
flow characteristics which leak de-anonymizing information. Because of the
uniformity of flowlets, no such information is leaked in \name for isolated
unidirectional flows. However, if the attacker is able to link the flowlets
corresponding to a bidirectional flow by their starting or ending time, then
an attack based on the RTT (see Section~\ref{sec:ta_passive}) could still be
possible. Such an attack can be thwarted by adding delays for setup packets and
flowlet start at the receiver, according to the path length (the shorter the
path, the longer the delay), as suggested by previous
work~\cite[Section~5.1]{Chen2015Hornet}.

\subsection{Defense against Active Attacks}
\label{sec:sub:active_attacks}

\paragraph{Tagging attacks}
A compromised node can modify packets adding tags that are recognizable by
downstream colluding nodes. This enables flow matching across flows observed at
different nodes~\cite{pries2008new}. \name defends against such attacks through
its per-hop packet authenticity (see Section~\ref{sec:specification}). A benign
node will detect and drop any modified packet.

\paragraph{Clogging attacks}
In clogging attacks, an adversary intentionally causes network
congestion~\cite{murdoch2005low, evans2009practical}, or
fluctuation~\cite{chakravarty2010traffic}
to create jamming or noticeable network jitter on relay nodes, and match
such patterns to deanonymize the path. Different from throughput fingerprint
attacks that aim to exert no influence on existing traffic patterns, clogging
attacks aggressively change the traffic patterns on victim links and are prone
to detection. First, clogging attacks in \name itself require DDoS capabilities
because of nodes' high bandwidth within the network.  In addition, \name nodes
attacked by clogging would run out of cached chaff packets, which in turn shuts
down the flowlet and prevents any additional matching. Moreover, given the large
number of flowlets in the network at any given time, the number of flowlets
terminated due to normal operations is large, which hides the fact that
the specific attacked flowlet is terminated.

\paragraph{Flow dynamics modification attacks}
Traffic pattern modulation attacks require attackers to modulate inter-packet
intervals to either create recognizable patterns (e.g., flow watermarking
attacks~\cite{houmansadr_rainbow:_2009, houmansadr_swirl:_2011}), or embed
identity information (e.g. flow fingerprinting
attacks~\cite{houmansadr2013}), so that downstream adversarial nodes can
deanonymize traffic by extracting the introduced traffic patterns.
Depending on the amount of perturbation introduced by the adversary, we can distinguish two
cases. In the first one, the adversarial actions fail to exhaust the cached
chaff packets on the node under attack for the target flowlet. In this case, the
outgoing rate for the flowlet at the node remains unchanged, and the attack is
ineffective. In the second case, the victim node runs out of cached chaff
packets for the target flowlet. In this case, the node terminates the flowlet to
prevent downstream nodes from observing the injected patterns.

 \subsection{Anonymity Set Size Evaluation}
 \label{sec:sub:ass}
 
 \carmela{I do like the analysis below, but if I understand correctly this
 	analysis explains how a given AS can learn from where a flow comes from given
 	the AS topology. This a fair analysis, but it seems like a partial adversary and
 	not a global one. So how does it relate to the threat model? \\
 	Then, I do not really see how it related to the use of batching in the setup
 	phase. Batching, if I got it correctly, happens at every router, right? does
 	that make any difference? If not, if the only thing the adversary can see are
 	inputs to ASs I am not sure how batching helps at any point.
 }
 
 \paragraph{Relationship anonymity set}
 Network-layer anonymity protocols are vulnerable to passive attacks based on
 network topology information launched by a single compromised AS. Compared to
 overlay-based anonymity systems~\cite{dms04} that allows global re-routing,
 traffic of network-layer anonymity protocols follows paths created by underlying
 network architectures. By observing the incoming and outgoing
 links of a packet, a compromised AS can derive network location
 information of communicating end hosts. For example, in
 Figure~\ref{fig:relationship_anonymity_example}, by forwarding a packet from AS1
 to AS3, AS2 knows that the sender must reside within the set \{AS0, AS1\} and
 the receiver falls into the set \{AS3, AS4, AS5\}. We name the former
 anonymity set \emph{sender anonymity set}, denoted as $S_s$, and
 call the latter anonymity set \emph{recipient
 	anonymity set}, denoted as $S_d$. Accordingly, we define \emph{relationship
 	anonymity set} $S_r = \{(s, d) | s \in S_s, \, d \in S_d\}$.
 
 To evaluate relationship anonymity of different protocols, we use anonymity-set size as the metric. By definition of $S_r$, the anonymity-set size $|S_r| = |S_s| \times |S_d|$. In Figure~\ref{fig:relationship_anonymity_example}, there are 8 hosts in both AS0 and AS1. Thus, $|S_s| = 16$. Similarly, we can calculate that $|S_d| = 24$ and $|S_r| = 16 \times 24 = 384$.

 Protocol designs influence corresponding anonymity-set sizes. In \lap and
 \dovetail, by analyzing header formats, a passive adversary can determine its
 position on the packet's path, i.e., its distances from the sender and the
 receiver~\cite{Hsiao2012, Sankey2014}. In
 Figure~\ref{fig:relationship_anonymity_example}, if the adversary in AS2 knows
 the sender is 2 hops away and the receiver is 1 hops away through analyzing
 packet headers, it can deduce that the sender must be in AS0 and the receiver
 must be in AS3. The resulting anonymity-set size is reduced to 8 * 8 = 64. In
 comparison, \name and \hornet's header designs prevent their headers from
 leaking position information.
 
 \paragraph{Experiment setup}
 We use a trace-based simulation to evaluate anonymity set sizes of different
 network-layer anonymity protocols in real world scenarios. We obtain the
 real-world AS-level topology from CAIDA AS relationship
 dataset~\cite{caida_as_rel}. We also annotate each AS with its IPv4 address
 space using the Routeview dataset~\cite{routeview}. In addition, we estimate
 real-world paths using iPlane traceroute datasets~\cite{iplane_dataset}. We use
 the traceroute traces on Dec. 12th, 2014. For each IP address--based trace, we
 convert it to AS path. Our preliminary analysis shows that the median AS path
 length is 4 and the average AS path length is 4.2. More than 99.99\% of AS paths
 have length less than 8.
 
 For each AS on a path in our path dataset, we compute the sizes of the
 relationship anonymity sets observed by the compromised AS in one of two
 scenarios: 1) the AS knows its position on path as in \lap and \dovetail;
         2) the AS has no information about its position on the path as in \hornet and \name. To compute
 anonymity set sizes, we first derive relationship anonymity sets composed by
 ASes. Then we compute the number of hosts in the ASes as the size of anonymity
 set size. We approximate the number of hosts within an AS by the number of IPv4
 addresses of that AS.

 \begin{figure*}[!btp]
 	\centering
 	\subfigure[Example scenario]{
 		 		\includegraphics[width=0.3\textwidth]{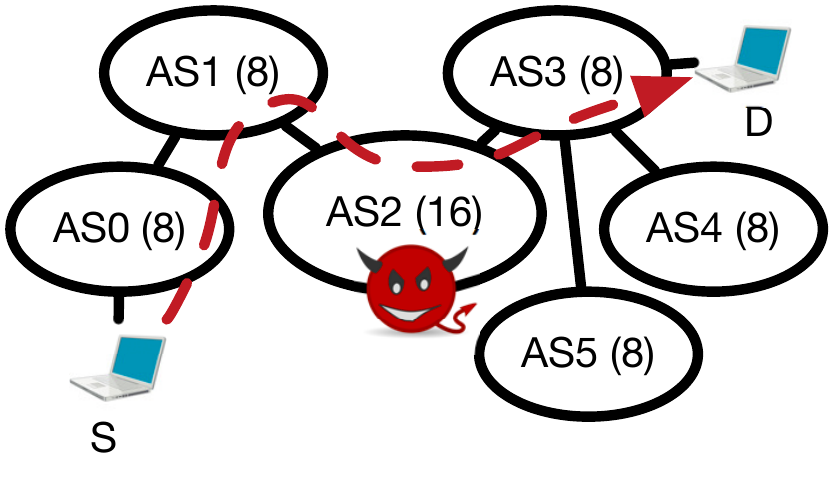}
 		\label{fig:relationship_anonymity_example}
 	}
 	\subfigure[\lap and \dovetail]{
 		\includegraphics[width=0.3\textwidth]{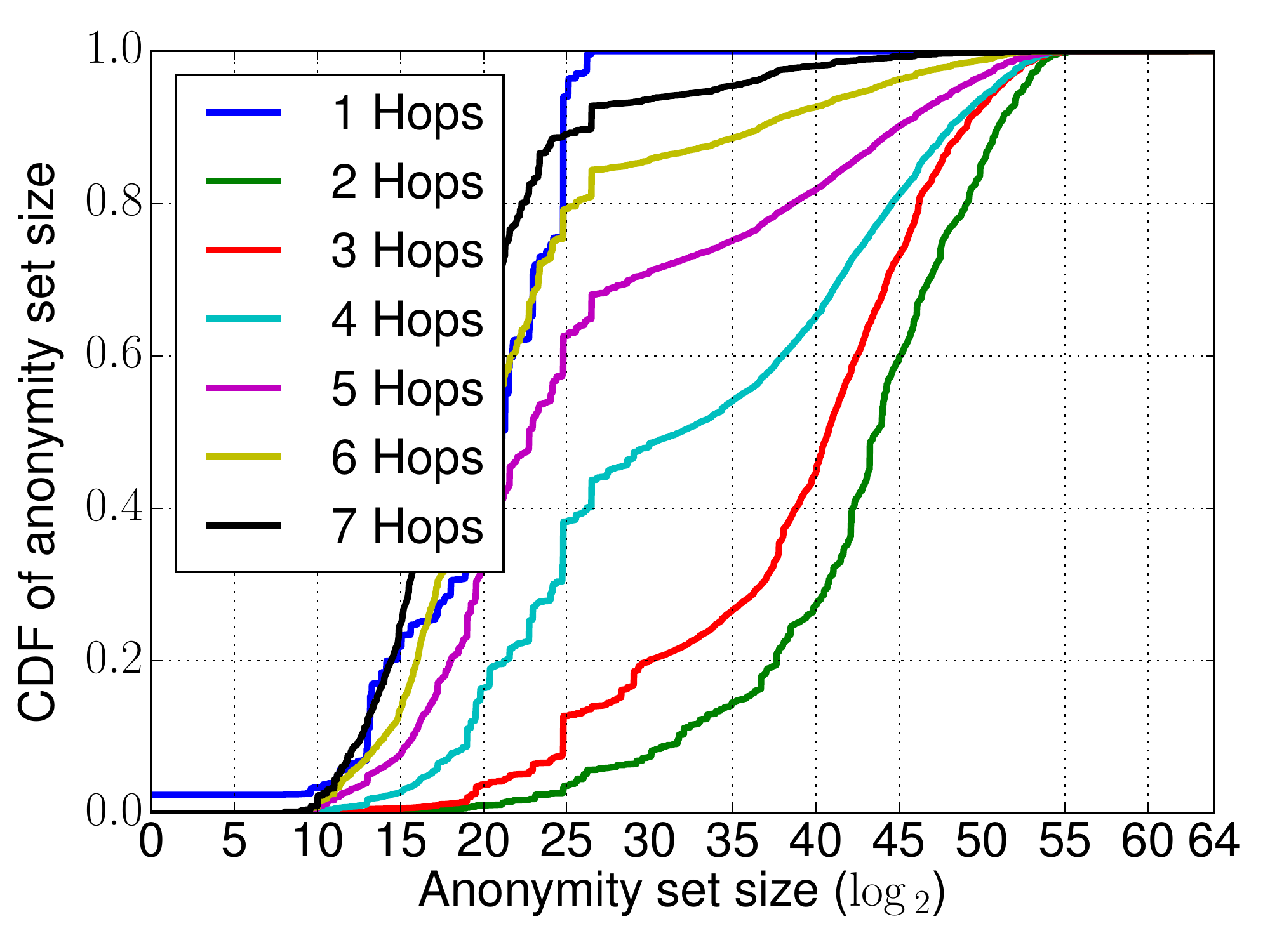}
 		\label{fig:anonymity_set_size_hop}
 	}
 	\subfigure[\name and \hornet]{
 		 		\includegraphics[width=0.3\textwidth]{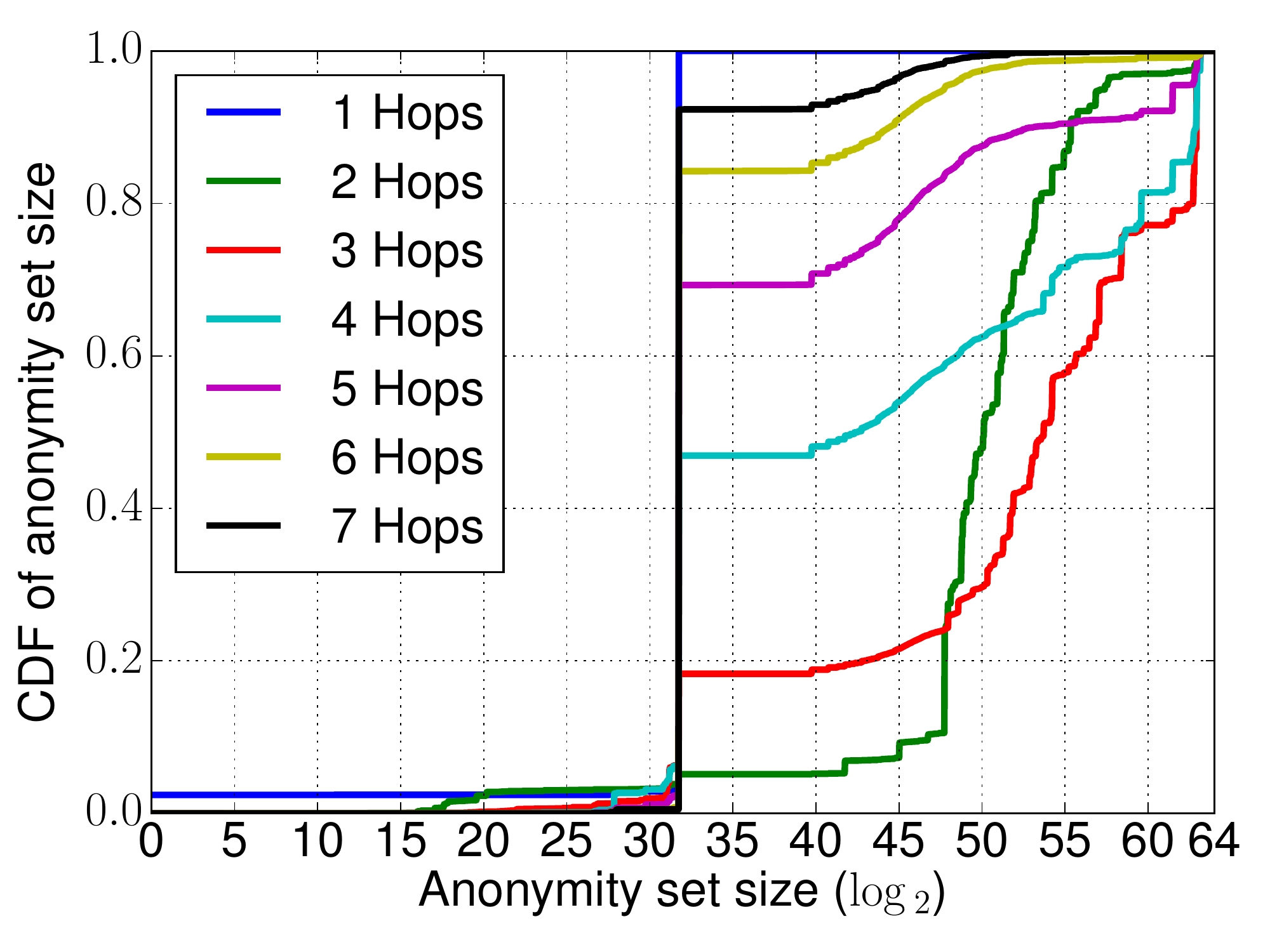}
 		\label{fig:anonymity_set_size_full}
 	}
 	\caption{a) A toy example of an adversary that exploits topology information
 		to de-anonymize a flowlet between sender $S$ and receiver $D$. AS$x$
 		($y$) denotes an AS with AS number $x$ and $y$ hosts attached. We assume
 		that the adversary compromised AS2. b) Cumulative Distribution Functions
 		(CDF) of anonymity set sizes for \lap and \dovetail. c) CDFs of anonymity
 		set sizes for \name and \hornet. In both b) and c), different lines
 		demonstrate anonymity set size distribution observed by an adversary that is
 		a fixed number of AS hops away from $S$.
 		 		 		 		 		 		 		 		 	}
 \end{figure*}
 
 \paragraph{Result}
 Figure~\ref{fig:anonymity_set_size_hop} demonstrates CDFs of anonymity-set sizes
 for \lap and \dovetail observed by a compromised AS.
 Figure~\ref{fig:anonymity_set_size_full} shows the CDF of anonymity-set sizes
 for \name and \hornet. In general, anonymity-set sizes of \name and \hornet
 exceed $2^{32}$ with probability larger than 95\% regardless
 of the adversary's on-path positions.
       The 90th percentiles of anonymity-set sizes of \name and \hornet are
 4--$2^{18}$ times larger than those of \lap and \dovetail depending on the
 distances between senders and receivers. We remark that when an AS is 6 or 7
 hops away from a sender, it is the last-hop AS with high probability, because
 99.99\% paths are less than 8 hops long.
     When the compromised ASes are 1 hop away
 from senders and when the ASes are close to receivers (6--7 hops away from
 senders), the gap between \name/\hornet and \lap/\dovetail is largest.

 \paragraph{Topology-based attacks and traffic analysis}
 In \lap, \dovetail, and \hornet, when an adversary compromises more than 
 1 AS on a path, he/she can correlate observation
 from different non-adjacent ASes by traffic analysis, such as flow fingerprint
 attacks~\cite{houmansadr2013}, to facilitate topology-based attacks.
 Assume that an adversary compromises $q$ ASes and observes a series of sender
 anonymity sets $\{S_s^i; i\in[1, q]\}$ and a series of recipient anonymity sets
 $\{S_d^i; i\in[1, q]\}$. The resulting relationship anonymity set size 
 $|S_r| = \min_{i\in[1, q]}|S_s^i| \times  \min_{i\in[1, q]}|S_d^i|$.
 For example, in Figure~\ref{fig:relationship_anonymity_example}, if the adversary
 compromises AS0 besides AS2 and correlates traffic from
 the same flowlet, the resulting relationship anonymity-set size $|S_r|$ is only
 $24$ ($1 \times 24$) compared to $384$ when only AS2 is compromised.
 
 \name improves over \lap, \dovetail, and \hornet by introducing defense against 
 traffic analysis (see Section~\ref{sec:sub:passive_attacks} and
 \ref{sec:sub:active_attacks}). By defeating traffic
 analysis and preventing correlation of flowlets at multiple non-adjacent
 ASes, \name enlarges the observed relationship anonymity set size.
 The resulting relationship anonymity set is only the smallest one 
 among the relationship anonymity sets observed by non-collaborative
 compromised ASes. $|S_r| = \min_{i\in[1, q]}|S_s^i| \times |S_d^i|$.
 For instance, when the adversary
 compromises AS0 besides AS2 and uses traffic analysis to correlate
 observed flowlets, the resulting relationship anonymity-set size $|S_r|$ 
 increased to $56$.
 
 \subsection{Formal Proof of Security}
 \label{sec:appendix:proof}
 
 Proof of \name security comes into two parts: the security proof of \name's setup phase protocol, and the security proof of \name's data transmission phase. 
 We derive the security of \name's setup phase by the security of the \sphinx protocol~\cite{Danezis2009},
 because \name's setup phase protocol directly uses the \sphinx protocol and Danezis and Goldberg have demonstrated that the \sphinx protocol realizes an ideal onion routing protocol defined by Camenisch and Lysyanskaya~\cite{camenisch2005formal}.
 
 In this section, we focus on the security of \name's data transmission phase and prove that \name's data transmission phase is equivalent to an ideal onion routing protocol based on UC framework~\cite{camenisch2005formal}. According to Camenisch and Lysyanskaya, a protocol is an ideal onion routing protocol if it offers four properties: \emph{correctness}, \emph{integrity}, \emph{wrap-resistance}, and \emph{security}. We briefly rephrase the definitions of four properties as follows:
 \begin{itemize}
 	\item \emph{Correctness}. The protocol should operate correctly without adversaries.
 	\item \emph{Integrity}. There exists an upper bound $N$ for the protocol, such that an adversary cannot forget a message that traverse more than $N$ hops in the network.
 	\item \emph{Wrap-resistance}. Given an output packet of an uncompromised node, an adversary cannot forget the corresponding input packet.
 	\item \emph{Security}. An adversary cannot distinguish among packets that enter an uncontrolled node $n$ even if the adversary is able to 1) select paths for the packets forwarded by $n$, 2) control all nodes on the path except $n$, and 3) observe all input and output packets of $n$ except the challenge packets.
 \end{itemize}
 
 \subsubsection{Correctness}
 A careful scrutiny of Section~\ref{sec:specification} should suffice to demonstrate the correctness of \name's data transmission protocol.
 
 \subsubsection{Integrity}
 We show that with significantly less than $2^k$ work, an adversary cannot forge a message ($IV_0$, $FS_0$, $\gamma_0$, $\beta_0$, $O_0$) that traverses more than $r$ hops nodes $n_0$, $n_1$, $\cdots$, $n_r$ in the network, even if the adversary learns all the secret keys $SV_0$, $SV_1$, $\cdots$, $SV_r$ for the nodes on the path. We construct a proof of contradiction.
 
 For convenience, we introduce a series of notations:
 \begin{align}
 \omega(x, y) &= \PRP^{-1}(h_{\PRP}(x); y)_{[0 .. k-1]}\\
 \rho(iv, x, y) &= \PRG(h_{\PRG}(\omega(x, y)\oplus iv))\\
 \sigma(iv, x, y) &= \PRP(h_{\PRP}(\omega(x, y)); iv)\\
 \tau(iv, x, y, o) &= \DEC(h_{\DEC}(\omega(x, y)); iv; o)
 \end{align}
 
 Assume that the adversary can create a message ($IV_0$, $FS_0$, $\gamma_0$, $\beta_0$, $O_0$) that traverses $n_0$, $n_1$, $\cdots$, $n_r$. We can rewrite 
 the message received by $n_r$, ($IV_r$, $FS_r$, $\gamma_r$, $\beta_r$, $O_r$),
 as follows:
 \begin{align}
 FS_r &= \bigoplus_{i=0} ^{r-1} \rho(IV_i, SV_i, FS_i)_{[c(r-i-1)+b .. 
 	c(r-i-1) + b + |FS| - 1]} \label{eq:proof:integrity:fs}\\
 \gamma_r & =  \bigoplus_{i=0} ^{r-1} \rho(IV_i, SV_i, FS_i)_{[c(r-i-1)+b + |FS|  .. 	c(r-i-1) + b + |FS| + k - 1]} \label{eq:proof:integrity:gamma}\\
 \beta_r &= \bigoplus_{i=1} ^{r-1} \rho(IV_i, SV_i, FS_i)_{[c(r-i) .. c r-1 ]}
 \label{eq:proof:integrity:beta}
 \end{align}
 where $IV_{i+1} = \sigma(SV_i, FS_i, IV_{i})$ and $O_i = \tau(IV_i, SV_i, FS_i, O_{i-1})$ $\forall\; 0 <i < r$.
 
 In order for the MAC to check on node $n_r$ so that $n_r$ will forward the packet, we need:
 \begin{equation}
 \gamma_r = \MAC(h_{\MAC}(\omega(SV_r, FS_r)); FS_r||\beta_r||O_r)
 \label{eq:proof:integrity:target}
 \end{equation}
 If we substitute Equation~\ref{eq:proof:integrity:fs}, \ref{eq:proof:integrity:gamma}, and \ref{eq:proof:integrity:beta} into Equation~\ref{eq:proof:integrity:target}, the right side of Equation~\ref{eq:proof:integrity:target} becomes a function with input ($IV_0$, $O_0$, $SV_0$, $\cdots$, $SV_r$, $FS_0$, $\cdots$, $FS_{r-1}$):
 \begin{align}
 \gamma_r &= \MAC(h_{\MAC}(\omega(SV_r, FS_r)); FS_r||\beta_r||O_r)\nonumber \\
 &= \mu(IV_0, O_0, SV_0, \cdots, SV_r, FS_0, \cdots, FS_{r}) \label{eq:proof:integrity:item4}
 \end{align}
 
 Before continuing, we first prove the following lemma:
 
 \begin{lemma}
 	With significantly less than $2^k$ work, an adversary can only distinguish $\mu$($IV_0$, $O_0$, $SV_0$, $\cdots$, $SV_r$, $FS_0$, $\cdots$, $FS_{r-1}$) from a random oracle with negligible probability.
 \end{lemma}
 
 \begin{proof}
 	We prove a statement equivalent to the lemma: with significantly less than $2^k$ work, an adversary cannot find two sets
 	\begin{align}
 	&(IV_0, O_0, SV_0, \cdots, SV_r, FS_0, \cdots, FS_{r-1}) \neq \nonumber \\
 	&(IV_0', O_0', SV_0', \cdots', SV_r', FS_0', \cdots, FS_{r-1}')
 	\end{align}
 	such that they lead to the same value of $\mu$. We will prove this by proof of contradiction. 
 	
 	Assume the adversary found two distinguished values that yields the same value of $\mu$. Because MAC is a random oracle, with significantly less than $2^k$ work, the attacker has to guarantee:
 	\begin{align}
 	\omega(SV_r, FS_r) &= \omega(SV_r', FS_r') \label{eq:proof:integrity:item1} \\ (FS_r||\beta_r||O_r)&=(FS_r'||\beta_r'||O_r')
 	\label{eq:proof:integrity:item2}
 	\end{align}
 	
 	Given the definition of $\omega$ and Equation~\ref{eq:proof:integrity:item1}, because \PRP is a pseudo-random permutation and $h_{\PRP}$ is collision resistant, the adversary must have $SV_r = SV_r'$.
 	
 	In addition, Equation~\ref{eq:proof:integrity:item2} determines $FS_r = FS_r'$, $\beta_r = \beta_r'$. We will show that the latter means 
 	\begin{align}
 	&(IV_0, SV_0, \cdots, SV_{r-1}, FS_0, \cdots, FS_{r-1}) = \nonumber\\
 	&(IV_0', SV_0', \cdots, SV_{r-1}', FS_0', \cdots, FS_{r-1})'\label{eq:proof:integrity:item3}
 	\end{align}
 	
 	Consider the last $c$ bits of $\beta_r$ and $\beta_{r}'$. By Equation~\ref{eq:proof:integrity:beta}, we have 
 	\begin{align}
 	&\rho(IV_{r-1}, SV_{r-1}, FS_{r-1})_{[c(r-1)..cr-1]} = \nonumber\\
 	&\rho(IV_{r-1}', SV_{r-1}', FS_{r-1}')_{[c(r-1)..cr-1]}
 	\end{align}
 	Because \PRG\xspace is a secure pseudo-random generator, the following equation holds
 	\begin{equation}
 	\omega(SV_{r-1}, FS_{r-1}) \oplus IV_{r-1} = \omega(SV_{r-1}', FS_{r-1}')  \oplus IV_{r-1}'
 	\end{equation}
 	Or
 	\begin{align}
 	&\omega(SV_{r-1}, FS_{r-1}) \oplus \sigma(SV_{r-2}, FS_{r-2}, IV_{r-2})
 	= \nonumber\\
 	&\omega(SV_{r-1}', FS_{r-1}') \oplus \sigma(SV_{r-2}', FS_{r-2}', IV_{r-2}')
 	\end{align}
 	Since  $\omega$ and $\sigma$ are two independent random oracles and their inputs do not overlap, the attacker has to ensure $SV_{r-1} = SV_{r-1}'$,
 	$FS_{r-1} = FS_{r-1}'$, $SV_{r-2} = SV_{r-2}'$, $FS_{r-2} = FS_{r-2}'$, 
 	and $IV_{r-2} = IV_{r-2}'$. 
 	
 	The equation that $IV_{r-2} = IV_{r-2}$ implies $FS_{r-3} = FS_{r-3}'$,
 	$SV_{r-3} = SV_{r-3}'$, and $IV_{r-3} = IV_{r-3}'$, because $\sigma$
 	is a random oracle. Repeating this logic, we will get Equation~\ref{eq:proof:integrity:item3}.
 	
 	Finally, given $O_r = O_{r}'$ and Equation~\ref{eq:proof:integrity:item3}, the attacker, with significantly less than $2^k$ work, has to make sure that
 	$O_0 = O_0'$.
 \end{proof}

 Let
 \begin{equation} 
 f(IV_i, SV_i, FS_i) = \rho(IV_i, SV_i, FS_i)_{[c(r-i)..cr-1]}\label{eq:proof:integrity:simple_rho}
 \end{equation}
 We now can substitute Equation~\ref{eq:proof:integrity:gamma} and \ref{eq:proof:integrity:simple_rho} 
 into Equation~\ref{eq:proof:integrity:item4} and rewrite the latter as:
 \begin{align}
 f(IV_0, SV_0, FS_0) = &\mu(IV_0, O_0, SV_0, \cdots, SV_r, FS_0, \cdots, FS_r) + \nonumber\\
 &\bigoplus_{i=1}^{r-1} f(IV_i, SV_i, FS_i)
 \end{align}
 Because MAC is not used in $f$, the right hand side of the above equation is also a random oracle, which we can denote as 
 \begin{equation}
 g(IV_0, O_0, SV_0, \cdots, SV_r, FS_0, \cdots, FS_r)
 \end{equation}
 
 To sum up, in order for the MAC $\gamma_r$ to check on node $n_r$, the attacker needs to find the solution to 
 \begin{equation}
 f(IV_0, SV_0, FS_0) = g(IV_0, O_0, SV_0, \cdots, SV_r, FS_0, \cdots, FS_r)
 \end{equation}
 with two independent random oracles $f$ and $g$. With significantly less than $2^k$ effort, the adversary can only succeed with negligible probability, which contradicts the assumption.

 \subsubsection{Wrap resistance}
 We prove that given a packet $(IV, FS, \gamma, \beta, O)$, an adversary cannot output a message $(IV', FS', \gamma', \beta', O')$ with significantly less than $2^k$ work, so that processing of the former packet on an uncompromised node leads to the latter one.
 
 If the adversary can succeed with significantly less than $2^k$ work, it is necessary that 
 \begin{equation}
 \rho(IV', SV, FS')_{[c(r-1)..cr-1]} = \beta_{[c(r-1)..cr-1]}
 \end{equation}
 
 Because $\PRP^{-1}$, $\PRG$, $h_{\PRP}$, and $h_{\PRG}$ are all random oracles and $SV$ is unknown to the adversary, with significantly less than $2^k$ work, the adversary can only succeed to generate correct values of ($IV'$, $FS'$) with negligible probability.
 
 \subsubsection{Security}
 To prove the security property, we construct the following game $G_0$. Given an uncompromised node $N$, the adversary selects two paths $n_0, n_1, \cdots, n_{\nu-1}$ ($0<\nu<=r$) and $n_0', n_1', \cdots, n_{\nu'-1}'$ ($0<\nu'<=r$) where $n_i = n_i'$ $\forall 0 \le i \le j$ and $n_j = n_j' = N$. The nodes following $N$ are not necessarily the same sets of the nodes and the length of two paths can be different. The adversary is also able to choose all secrets for all nodes except for N, including the public/private keys and local secrets. Moreover, the adversary can also arbitrarily decide the contents of payload $O$.
 
 The challenger randomly selects a bit $b$ and proceeds in one of the two following ways:
 
 $b=0$. The challenger establishes a flowlet through the path $n_0, \cdots,
 n_{\nu-1}$ and then creates a data packet with payload $O$ chosen by the
 adversary. The challenger outputs $(IV_0, FS_0, \gamma_0, \beta_0, O_0)$, which
 can be sent to node $n_0$. We use $(IV_i, FS_i, \gamma_i, \beta_i, O_i)$ to
 represent the corresponding packet received by node $n_i$ ($n_i'$) on the path. 
 
 $b=1$. The challenger establishes a flowlet through the alternate path $n_0',
 \cdots, n_{\nu'-1}'$ and outputs a data packet ($IV_0$, $FS_0$, $\gamma_0$,
 $\beta_0$, $O_0$) that can be sent to $n_1$.
 
 Given the output $(IV_0, FS_0, \gamma_0, \beta_0, O_0)$, the adversary is challenged to determine $b$. The adversary can additionally input up to $q$ messages so long as they are not equal to $(IV_j, FS_j, \gamma_j, \beta_j, O_j)$.
 
 The adversary's advantage is defined as 
 $$adv = |Pr(success) - \frac{1}{2}|$$ 
 We will show that 
 the advantage is negligible with less than $2^k$ work.
 
 \begin{proof}
 	We use a hybrid-game method by establishing the following two new games $G_1$ and $G_2$. The definition of $G_1$ is the same as $G_0$ except that we require that $N=n_0=n_0'$, i.e., the first node is uncompromised. We further define $G_2$, whose assumption is the same as $G_1$ with only one exception that $(IV_0, FS_0, \gamma_0, \beta_0, O_e)$ are randomly drawn from the corresponding domains.
 	
 	First, because in $G_2$, the message are all randomly drawn, the adversary's advantage in guessing the bit is 0. Next, we would show that in a chain of game $G_0 \rightarrow G_1 \rightarrow G_2$, the adversary can only distinguish a game from the previous game with negligible success probability with significantly less than $2^k$ work.
 	
 	$G_0 \rightarrow G_1$. On one hand, it is obvious that an adversary who can succeed in $G_0$ is able to succeed in $G_1$ as the former one is a more general game. On the other hand, because the adversary fully control nodes $n_0, \cdots, n_{j-1}$ and can thus emulate their packet processing, the adversary can win game $G_0$ if s/he can win $G_1$.
 	
 	$G_1 \rightarrow G_2$. In order to distinguish $G_2$ from $G_1$, the following statements must be true:
 	\begin{itemize}
 		\item The adversary can distinguish $FS_0$ $=$ $\PRP$($h_{\PRP}(SV_0)$; $s_0 || R_0$) from randomness without knowing $SV_0$. Because $\PRP$ is pseudo-random permutation with security parameter $k$, the probability that the adversary succeeds with less than $2^k$ work is negligible.
 		\item The adversary can distinguish 
 		\begin{align}
 		\beta_0 = \PRG(h_{\PRG}(SV_0) \oplus & IV_0)_{[c(r-1) .. cr - 1]} \oplus \nonumber \\
 		&{ctrl_1 || FS_1 || \gamma_1 || \beta_1)}
 		\end{align}
 		from random bits. Because $\PRG$ is secure pseudo-random number generator, and the adversary has no knowledge of $SV_0$, the probability that the adversary succeeds is negligible.
 		\item We can repeat the same logic to show that it is impossible to distinguish $\gamma_0$ and $O_0$ from random bits with non-negligible probability with significantly less than $2^k$ work.
 	\end{itemize}
 \end{proof}

\section{Evaluation}
\label{sec:evaluation}
This section describes our implementation of \name{}, a performance evaluation,and our evaluation of bandwidth overhead added by end-to-end traffic shaping.

\subsection{Implementation on High-speed Routers}
\label{sec:implementation}
We implement \name's setup and data
transmission logic on a software router. We use Intel's Data Plane
Development Kit (DPDK~\cite{dpdk}, version 2.1.0), which supports fast packet processing in user space. We assemble a customized cryptography library based
on the Intel AESNI sample library~\cite{aesnilibrary}, and the curve25519-donna~\cite{curve25519donna} and PolarSSL~\cite{polarssl} libraries.  We use 128-bit AES counter mode for encryption and 128-bit AES CBC-MAC.

\begin{figure*}[tbp]
	\centering
	\subfigure[Latency]{
		\includegraphics[width=0.3\textwidth]{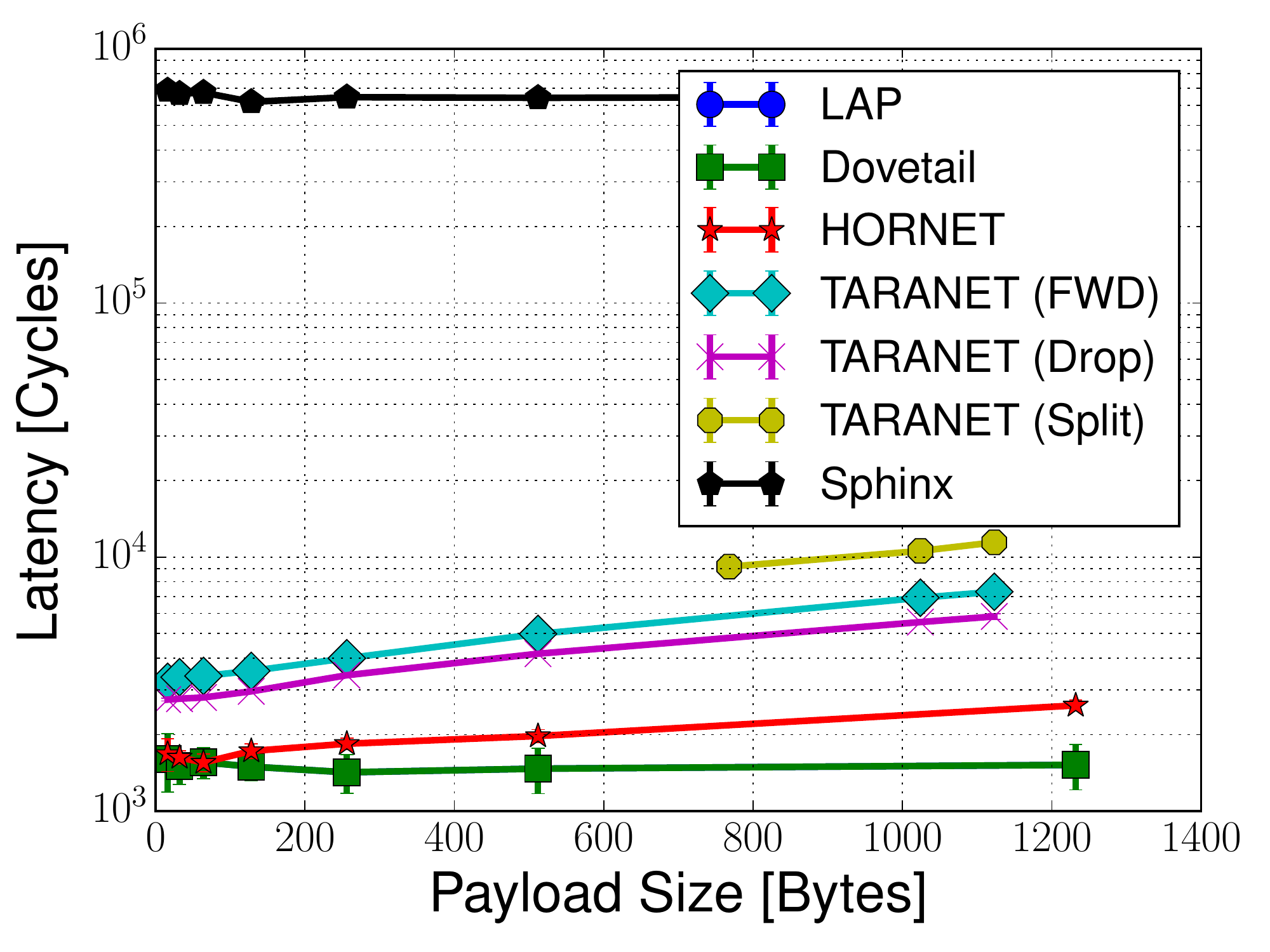}
		\label{fig:latency}
	}
	\subfigure[Goodput: 7-hop header]{
		\includegraphics[width=0.3\textwidth]{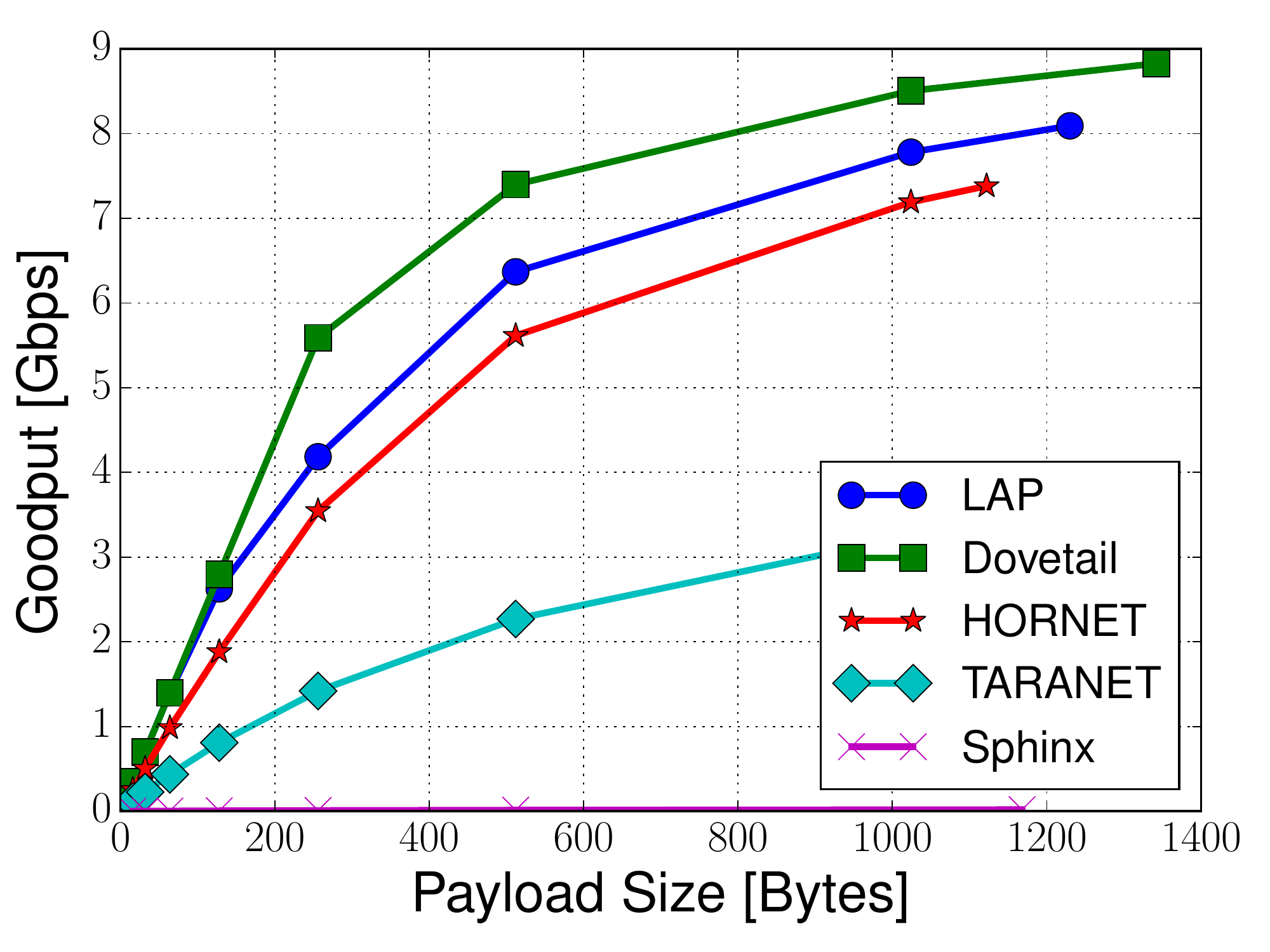}
		\label{fig:goodput_small}
	}
	\subfigure[Goodput 14-hop header]{
		\includegraphics[width=0.3\textwidth]{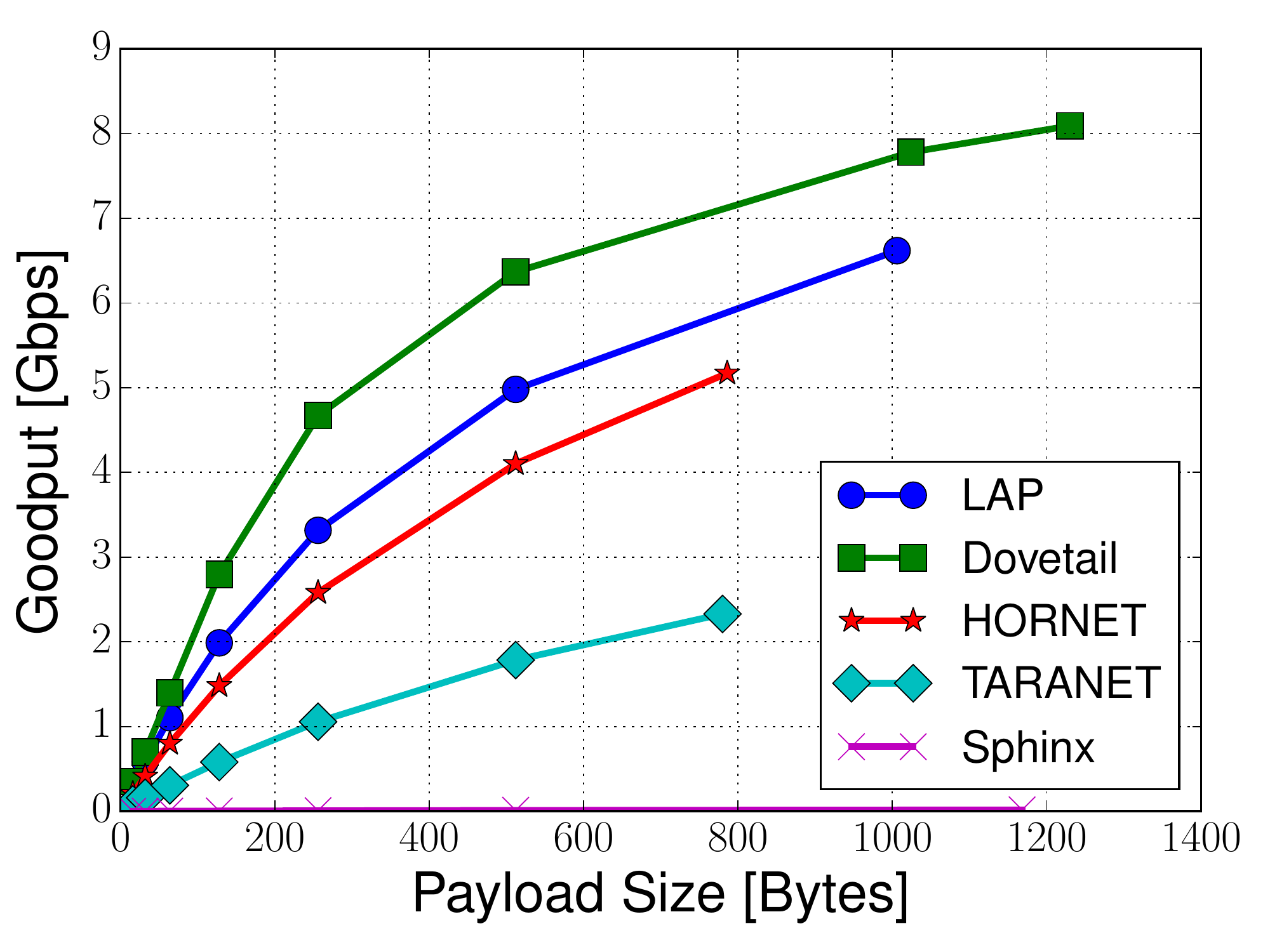}
		\label{fig:goodput_large}
	}
	\caption{a) Average latency of processing a packet for different protocols with error bars (95\% confidence intervals). For a packet with ``SPLIT'' flag, because the payload has to contain at least two other packet headers, we only test packets with payloads at least 768~bytes. Lower is better. b) Data forwarding goodput on a 10~Gbps link for packets with 7-hop headers and different payload sizes; c) data forwarding goodput on a 10~Gbps link for packets with 14-hop headers and different payload sizes. Higher is better.}
\end{figure*}

\subsection{Performance Evaluation}
Our testbed is composed of a commodity Intel server and a Spirent TestCenter
packet generator~\cite{spirent}. The Intel server functions as a software
router and is equipped with an Intel Xeon E5-2560 CPU (2.70~GHz, 2 sockets, 8
cores per socket) and 64~GB DRAM. The server also has 3 Intel 82599ES
network cards with 4 ports per card, and is connected to the packet generator
through twelve 10-Gbps links. Thus the testbed can test throughput up
to 120~Gbps.

To remove implementation bias and allow fair comparison with other
anonymity protocols, we additionally implement \lap~\cite{Hsiao2012},
\dovetail~\cite{Sankey2014}, \hornet~\cite{Chen2015Hornet}, and
\sphinx~\cite{Danezis2009} logic using our custom cryptography library and
DPDK. Note that \lap, \dovetail, and \hornet are high-speed network-layer anonymity
protocols but cannot defend against traffic analysis attacks. \sphinx is a mix network that requires performing public key cryptographic operations for every data packet and incurs high computation overhead. 

\chen{the following paragraph reminds the reviewers that \name should be slower than \lap,
	\dovetail, and \hornet, that \name should be faster than \sphinx, and that \name should provide high throughput. We do not want the reviewers to have false hopes.}
We remind the reader that \name's performance is lower than \lap, \dovetail, and \hornet, because 
\name's traffic-analysis resistance property incurs additional overhead by design. 
However, \name should outperform \sphinx, which also offers traffic-analysis resistance. Additionally,
as an essential requirement of high-speed network-layer anonymity protocol, we expect that \name should sustain high forwarding throughput.

\paragraph{Processing latency} We first evaluate the average latency of processing a data packet on a single core
using different anonymity protocols. For \name, we also compare the latency of
performing different mutation actions. The results are shown in Figure~\ref{fig:latency}.

\name's processing latency is comparable to \lap, \dovetail, and \hornet.
When the payload size is smaller than 64~bytes, processing a \name data packet (following the steps described above) incurs less than 1$\mu$s ($\approx$3700 cycles) per-hop overhead on a single core.
For payloads larger than 1024~bytes, the latency increases to up to 2$\mu$s
($\approx$7200 cycles). Splitting a \name packet incurs only additional 1$\mu$s ($\approx$4200 cycles). Since the total number of ASes on
a path is usually less than 7~\cite{Chen2015Hornet}, \name processing will
add only $\sim$20$\mu$s to the end-to-end latency.

Processing a setup packet on our test machine incurs around 250$\mu$s
(0.66M cycles) per hop per packet.  This is due to setup packets
requiring a DH key-exchange operation. However, for path lengths of less than 7
hops, this latency adds less than 2ms at the start
of each flowlet.

\paragraph{Goodput}
Goodput measures the throughput of useful data that can be transmitted by the protocol as it separates data throughput and packet header overhead. Figures~\ref{fig:goodput_small} and \ref{fig:goodput_large} show the goodput of
different protocols with 7-hop and 14-hop headers, respectively, on a single 10~Gbps link with 1 core assigned. We observe that even with longer processing latency and larger headers,
\name still achieves $\approx$45\% of \hornet's goodput in both cases. With a single core, \name can still achieve $\sim$0.37~Mpkt/s.

\begin{figure*}[tbp]
	\centering
	\subfigure[Split Rate: Different Success Rate]{
		\includegraphics[width=0.3\textwidth]{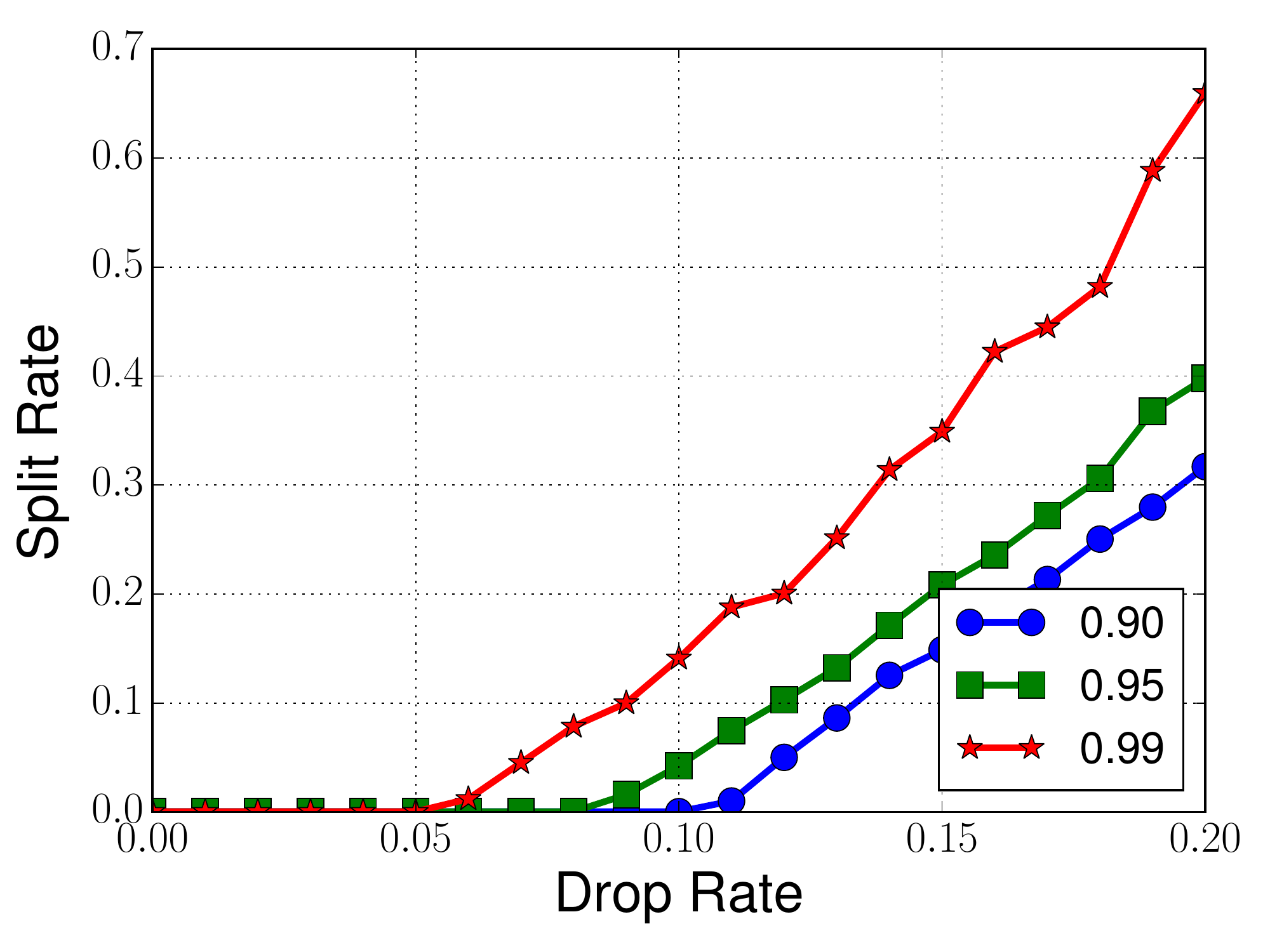}
		\label{fig:split_rate_succ_rate}
	}
	\subfigure[Split Rate: Different $H$]{
		\includegraphics[width=0.3\textwidth]{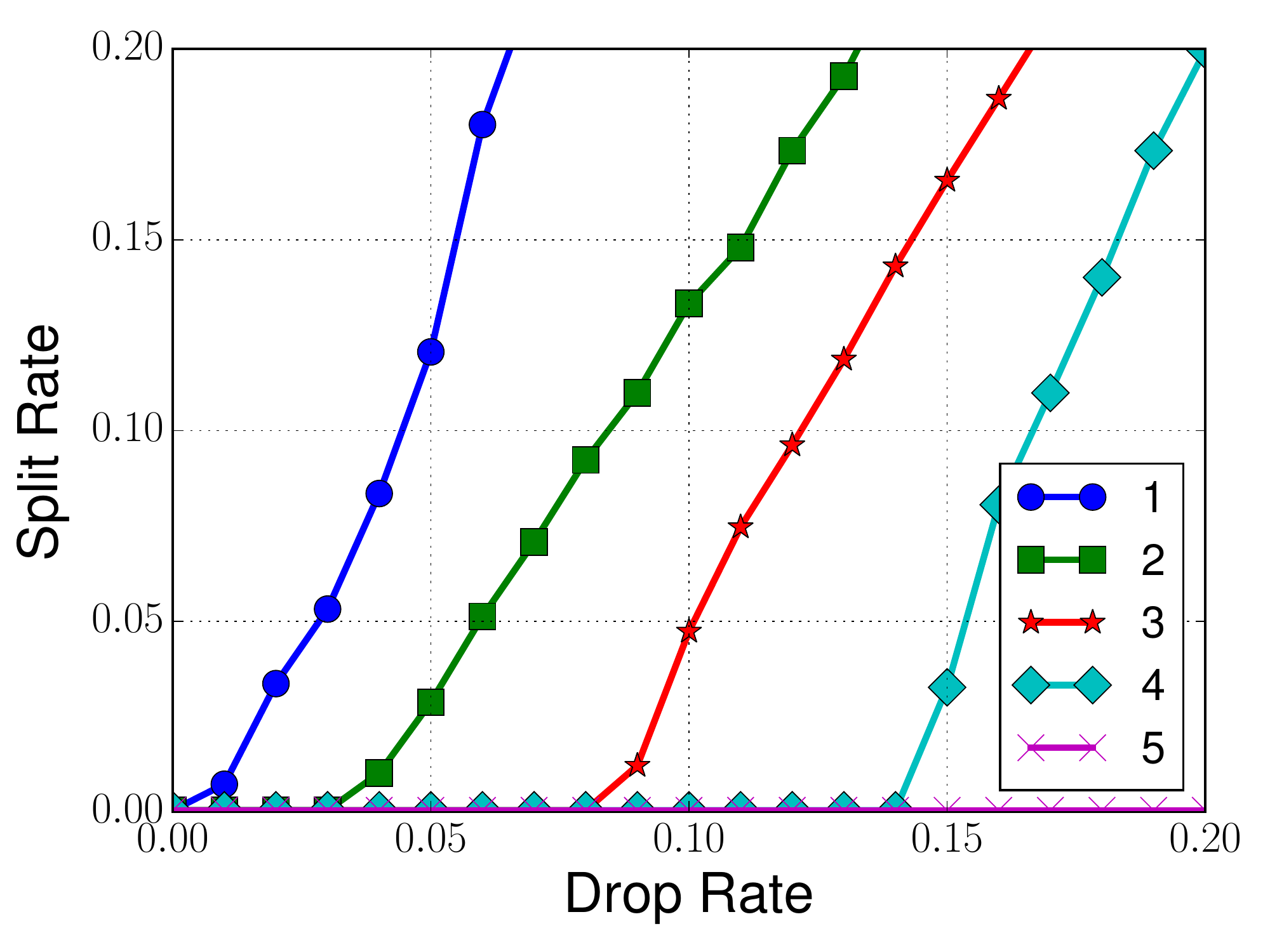}
		\label{fig:split_rate_hole_count}
	}
	\subfigure[Overhead Caused by Chaff Traffic]{
		\includegraphics[width=0.3\textwidth]{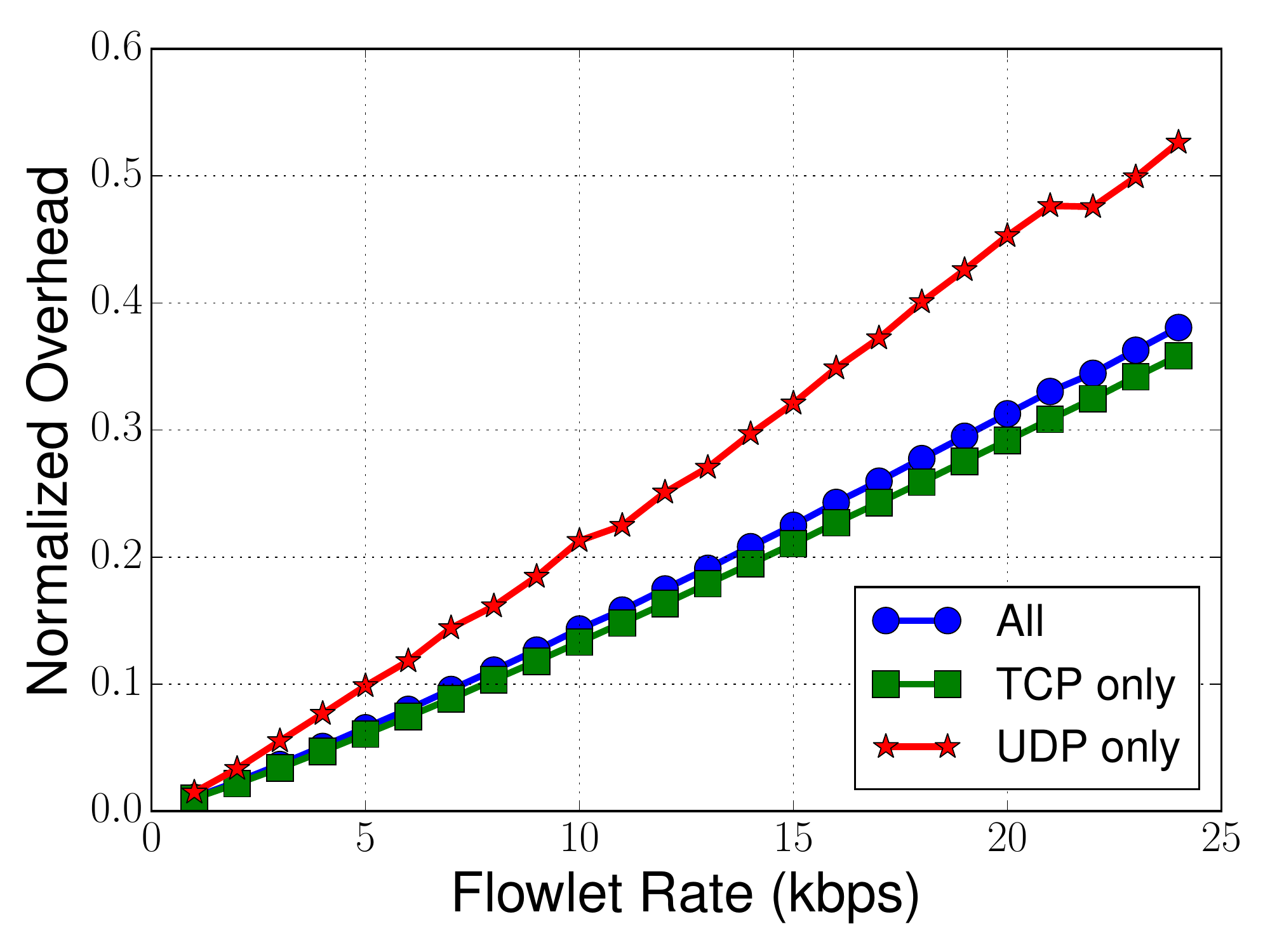}
		\label{fig:overhead_chaff}
	}
	\caption{a) Probability of splittable packets to guarantee different success rate. Failure counter $H=2$. Different lines stand for different success rates. b) Probability of splittable packets for different parameter $H$ (See Section~\ref{sec:flowlet}) to achieve a 95\% success rate. c) Bandwidth overhead caused by added chaff traffic for traffic shaping. We normalize the overhead by dividing the added overhead by the original bandwidth. Lower is better. }
\end{figure*}

\paragraph{Maximum total throughput}
To evaluate the maximum total throughput of our protocol with respect to the number of cores, we test \name with all twelve 10~Gbps ports enabled while using all 16 CPU cores. Each port has 1 input queue and 1 output queue.
To fully distribute the computation power of 16 cores to 12 input and 12 output queues, we assign 8 cores exclusively to 8 input queues, 4 cores each to one input and one output queue, and the remaining 4 cores to 8 output queues.
The packet generator generates packets that have random egress ports and saturate all 12 ports.
Our evaluation finds that \name can process anonymous traffic at 50.96~Gbps on our software router for packets with 512~bytes of payload, which is comparable to the switching capacity of a commercial edge router~\cite{asr1000}.

\paragraph{Delay of flowlet setup}\label{sec:setupevaluation}
For the setup phase, \name uses packet batching and randomization to protect
against traffic analysis. Our observation is that if the number of flowlet
setups is sufficient large, batching setup packets can still end up yielding a
short setup delay.

We conduct a trace-driven simulation using the CAIDA's anonymized packet traces to evaluate the setup phase's delay. The packet traces are recorded by the ``equinix-chicago'' monitor on a Tier-1 ISP's 10~Gbps link between 1-2 pm on Mar. 20th, 2014~\cite{caida-passive2014}
 We assume the first packet in each flow in the dataset is a flowlet setup
 packet. We simulate the latency introduced by batching, randomizing, and
 cryptographic processing by injecting the setup packet trace into a \name node
 and varying the batch size.

The resulting latency of flowlet setups increases almost linearly as the batch
sizes increases. When the batch size is 16, the per-hop latency is 1.6 $\pm$
1.0~ms (95\% confidence interval). When the batch sizes reaches 128, the per-hop
latency increases up to 12 $\pm$ 7~ms. For a path with 7 AS-hops that means the
flowlet setup will introduce less than $\sim$170~ms additional round-trip
latency for the setup phase, which is a small proportion of the delay of an
inter-continental path.

\subsection{Overhead Evaluation}
We conduct a trace-based evaluation of \name to evaluate added bandwidth overhead for end hosts and the amount of state for routers. For bandwidth overhead, we evaluate the number of splittable packets required to accommodate different levels of packet drops and the number of chaff traffic needed to shape real-world traffic.
We use CAIDA's anonymized packet traces as discussed in Section~\ref{sec:setupevaluation}. We filter away ICMP packets and small flows that has the size smaller than 10 packets or has the transmission rate lower than 1 byte/second. 

\paragraph{Split rate}
First, we evaluate the number of splittable packets needed to account for packet drops. Note that an insufficient split rate causes a node to deplete its locally cached chaff and prematurely terminate a flowlet when there exist large number of packet drops. We convert each flow in the trace into flowlets where $B=10$~kbps and $T=1$~min and run the trace with different per-hop split rates, different drop rates, and various failure counters $H$. We set $L_{chf}=3$ to let a node to cache 3 chaff packets at maximum for each flowlet.  Figure~\ref{fig:split_rate_succ_rate} shows the required per-hop split rates regarding different drop rates to achieve different success rates. The observed drop rate in the Internet is around 0.2\%~\cite{Sundaresan2011}. For such a low drop rate and a failure counter $H=2$, a node can set the per-hop split rate to almost 0 and still obtain high success rates as much as 99\%.  
To account for a highly lossy network link or an adversary that manipulates timing pattern through dropping packets,
an end host can adjust the split rate up to 5\% even for a very high per-hop drop rate 10\% to achieve 95\% success rates. 

Figure~\ref{fig:split_rate_hole_count} demonstrates the required per-hop split rates with respect to different drop rates to guarantee a 95\% success rate when the failure counter $H$ ranges from 1 to 5. In general, given a certain packet drop rate, a larger $H$ helps reduce the per-hop split rate. For instance, when $H$=1, the per-hop split rate can be as high as 12\% for a packet drop rate of 5\%. However, when $H$ increases to 2, the required per-hop split rate is already at 3.4\%. When we lets $H$=4, we can accommodate a per-hop drop rate of 15\% by a per-hop split rate as small as 3.7\%.
We remark that when flowlets have smaller bandwidth $B$ (as will be shown next), the success rate increases. 

\paragraph{Chaff overhead}
We then evaluate the bandwidth overhead of the added chaff traffic. We convert real-world flows in the CAIDA's packet traces into flowlets and compute the amount of chaff required. 
Note that we normalize the resulting overhead through dividing it by the total traffic volumes, in order to remove the
impact of traffic volumes and network sizes.

Figure~\ref{fig:overhead_chaff} plots the chaff overhead needed for the conversion when the bandwidth parameter of flowlets $B$ varies and $T$=1~min. Generally, large $B$ results in large chaff overhead. When we use $B$=5~kbps, the overhead of chaff packets is 7\%. In comparison, when $B$ becomes 20~kbps, the overhead of chaff traffic is increased to 31\%. 
Moreover, we observe that small flow sizes, such as UDP flows for DNS lookups, lead to large chaff overhead given the same $B$ because more packets in the resulting flowlets are chaff packets. In Figure~\ref{fig:overhead_chaff}, the chaff overhead for UDP flows is larger than for TCP flows because the size of UDP flows is usually larger than the size of TCP flows. 

\paragraph{Required amount of state and scalability}
To enable traffic shaping for flowlets, an intermediate node maintains state bounded by the node's bandwidth (Section~\ref{sec:flowlet}). To demonstrate the scalability of \name with respect to the Internet traffic volumes, we evaluate the amount of state required for a node to process real Internet traffic. 

For each flowlet that consumes bandwidth $B$, a node maintains $L_{chf}$ chaff packets and a failure counter required by the end-to-end traffic shaping technique (Section~\ref{sec:flowlet}).
We set $L_{chf} = 3$ and vary the bandwidth parameter $B$ for flowlets. Using the flowlets converted from our CAIDA flow trace, we can evaluate the amount of state required. Our results show that a node stores 90~MB state when $B$=10~kbps for a 10~Gbps link and that the node needs to store 52~MB state when $B$=20~kbps.

\section{Discussion}
\label{sec:discussion}

\subsection{Incremental Deployment}
\paragraph{Deployment Incentives}
We envision that ISPs have incentives to deploy \name to offer strong
anonymity as a service to their privacy-sensitive users or other
customer ISPs who in turn desire to offer anonymous communication
services. This would give \name-deploying ISPs a competitive
advantage: both private and business customers who want to use an
anonymity service would choose an ISP that offers privacy protection.

\paragraph{Incremental Deployment Strategy}
The minimal requirement for deploying \name includes a topology server
that distributes path information, a few ISPs that deploy border
routers supporting the \name protocol, and end hosts that run \name
client software. We remark that the network architectures that we
consider, such as NIRA~\cite{Yang2007NIRA},
NEBULA~\cite{Anderson2013NEBULA} and SCION~\cite{Xin2011SCION},
already assume such topology servers as part of necessary
control-plane infrastructure.

Admittedly, the more ISPs that deploy \name would increase the
anonymity set size which in turn benefits all users. However, a few
initial \name-enabled ISPs that share no physical links can establish
tunnels between each other through legacy ISPs and start to carry
anonymous traffic among users. As more \name-capable ISPs join the
\name network, tunnels are increasingly replaced with direct
ISP-to-ISP connections, which provides increasingly better guarantees.

\subsection{Limitations}
\paragraph{Long-term Intersection Attack}
An adversary who observes presence of all sender
and receiver clients over a long period of time can perform intersection
analysis~\cite{danezis2005, mathewson2005} to reveal pairs of clients that are
repeatedly online during the same period. Clients can minimize their risk by
being online not only when they are actively communicating, but in general this
attack is difficult to defend against through technical means. For further
defenses, \name could be enhanced using existing solutions,
such as dummy connections~\cite{berthold2003}, or with the Buddies system~\cite{wolinsky2013}, 
which allows clients to control which subset of pseudonyms appears online for a particular session. We leave
analysis and evaluation of integration with such systems to future work.

\paragraph{Routing Attacks}
\name relies on underlying network architectures for routing packets.
Adversarial nodes can attack underlying network architectures to place
themselves at strategic positions to launch traffic
analysis~\cite{vanbever2014anonymity, sun2015raptor}. Although defeating routing
attacks itself is beyond \name's scope,
the network architecture candidates we consider
offer control and integrity of packets' paths that prevent routing attacks. For
example, SCION~\cite{Xin2011SCION} and NEBULA~\cite{Anderson2013NEBULA} both
embed integrity tags within paths to prevent path modification. Pathlet~\cite{godfrey2009pathlet} pushes path selection to end
hosts, enabling end hosts to select the traversing path.

\paragraph{Denial-of-Service Attacks}
An adversary can initiate a high volume of flowlets passing through a node, to
exhaust the node’s computation power, bandwidth, and memory. \name itself cannot
defend against such a DoS attack. To mitigate the DoS attack, a node can require
flowlet initiators to solve cryptographic puzzles [26]. Additionally, an ISP
that operates \name can also directly restrict the average flowlet-initiation
rate of its customers. We note that a DoS attack aiming to exhaust a \name's
memory state by initiating large number of flowlets will fail, because the
amount of state that a \name node maintains is strictly linear to the node's
actual bandwidth. If the adversary is able to accumulate sufficient bandwidth,
such an attack will only jam the network link of the victim node and become
a bandwidth-based DoS attack that can be mitigated.

\section{Conclusion}
\label{sec:conclusion}
In this paper, we have shown that it is possible to
obtain the efficiency of onion-routing-based anonymity systems and the security
of mix-based systems while avoiding their disadvantages. We have
designed \name, which uses mixing and coordinated traffic shaping
to thwart traffic analysis for the setup phase and the data transmission phase respectively.  To achieve high
performance and scalability, we build on the key observation that high-speed
networks process enough volume that mixing at their core routers has minimal
performance overhead. The performance and security properties achieved by our
protocol suggest that efficient traffic-analysis-resistant protocol at the network
layer is feasible, and that the increased security warrants the additional performance
cost.
 
 \section{Acknowledgments}
 \label{sec:ack}
 We would like to thank the anonymous reviewers for their suggestions for improving the paper. We also appreciate the insightful discussion with the 
members of ETH Zurich Network Security group.

The research leading to these results has received funding from the
European Research Council under the European Union's Seventh Framework
Programme (FP7/2007-2013) / ERC grant agreement 617605. We gratefully
acknowledge support from ETH Zurich and from the Zurich Information
Security and Privacy Center (ZISC).
  
\bibliographystyle{plain}
{\tiny
	\balance \bibliography{bib.bib}

\begin{thebibliography}{10}

\bibitem{caida_as_rel}
{CAIDA AS-relationship dataset}.
\newblock \url{http://www.caida.org/data/as-relationships/}.

\bibitem{asr1000}
Cisco {ASR}-1000.
\newblock
  \url{http://www.cisco.com/c/en/us/products/routers/asr-1000-series-aggregation-services-routers/index.html}.

\bibitem{curve25519donna}
{curve25519-donna}.
\newblock \url{https://code.google.com/p/curve25519-donna/}.

\bibitem{dpdk}
{DPDK: Data Plane Development Kit}.
\newblock \url{http://dpdk.org/}.

\bibitem{aesnilibrary}
{Intel AESNI Sample Library}.
\newblock
  \url{https://software.intel.com/en-us/articles/download-the-intel-aesni-sample-library}.

\bibitem{iplane_dataset}
{iPlane dataset}.
\newblock \url{http://iplane.cs.washington.edu/data/data.html}.

\bibitem{polarssl}
{PolarSSL}.
\newblock \url{https://polarssl.org/}.

\bibitem{routeview}
Routeview project.
\newblock \url{http://www.routeviews.org/}.

\bibitem{segment_routing}
Segment routing architecture ({IETF} draft).
\newblock
  \url{https://datatracker.ietf.org/doc/draft-ietf-spring-segment-routing/}.
\newblock Retrieved on January 27, 2016.

\bibitem{spirent}
{Spirent TestCenter}.
\newblock
  \url{http://www.spirent.com/~/media/Datasheets/Broadband/PAB/SpirentTestCenter/STC_Packet_Generator-Analyzer_BasePackage_datasheet.pdf}.

\bibitem{caida-passive2014}
{The CAIDA UCSD Anonymized Internet Traces 2014}.
\newblock \url{http://www.caida.org/data/passive/passive_2014_dataset.xml}.

\bibitem{torusers}
Tor metrics: Direct users by country.
\newblock "\url{https://metrics.torproject.org/userstats-relay-country.html}.
\newblock Retrieved on Nov.3, 2015.

\bibitem{Anderson2013NEBULA}
Tom Anderson, Ken Birman, Robert Broberg, Matthew Caesar, Douglas Comer, Chase
  Cotton, Michael~J Freedman, Andreas Haeberlen, Zachary~G Ives, Arvind
  Krishnamurthy, et~al.
\newblock The nebula future internet architecture.
\newblock In {\em The Future Internet}, pages 16--26. Springer, 2013.

\bibitem{berthold2001web}
Oliver Berthold, Hannes Federrath, and Stefan K{\"o}psell.
\newblock Web mixes: A system for anonymous and unobservable internet access.
\newblock In {\em PETS}, 2001.

\bibitem{berthold2003}
Oliver Berthold and Heinrich Langos.
\newblock Dummy {Traffic} against {Long} {Term} {Intersection} {Attacks}.
\newblock In {\em {PETS}}, 2003.

\bibitem{blum2004detection}
Avrim Blum, Dawn Song, and Shobha Venkataraman.
\newblock Detection of interactive stepping stones: Algorithms and confidence
  bounds.
\newblock In {\em Recent Advances in Intrusion Detection}. Springer, 2004.

\bibitem{camenisch2005formal}
Jan Camenisch and Anna Lysyanskaya.
\newblock A formal treatment of onion routing.
\newblock In {\em CRYPTO}, 2005.

\bibitem{chakravarty2014}
Sambuddho Chakravarty, Marco~V. Barbera, Georgios Portokalidis, Michalis
  Polychronakis, and Angelos~D. Keromytis.
\newblock On the effectiveness of traffic analysis against anonymity networks
  using flow records.
\newblock In {\em {PAM}}, 2014.

\bibitem{chakravarty2010traffic}
Sambuddho Chakravarty, Angelos Stavrou, and Angelos~D Keromytis.
\newblock Traffic analysis against low-latency anonymity networks using
  available bandwidth estimation.
\newblock In {\em ESORICS}, 2010.

\bibitem{chaum1988dining}
David Chaum.
\newblock The dining cryptographers problem: Unconditional sender and recipient
  untraceability.
\newblock {\em Journal of cryptology}, 1(1), 1988.

\bibitem{chaum1981untraceable}
David~L Chaum.
\newblock Untraceable electronic mail, return addresses, and digital
  pseudonyms.
\newblock {\em Communications of the ACM}, 24(2), 1981.

\bibitem{Chen2015Hornet}
Chen Chen, Daniele~E. Asoni, David Barrera, George Danezis, and Adrian Perrig.
\newblock {HORNET}: High-speed onion routing at the network layer.
\newblock In {\em ACM CCS}, 2015.

\bibitem{Chor1998}
Benny Chor, Oded Goldreich, Eyal Kushilevitz, and Madhu Sudan.
\newblock Private information retrieval.
\newblock {\em Journal of the ACM}, 45(6), 1998.

\bibitem{danezistrafficanalysis2004}
George Danezis.
\newblock The traffic analysis of continuous-time mixes.
\newblock In {\em PETS}, 2004.

\bibitem{Danezis2009}
George Danezis and Ian Goldberg.
\newblock Sphinx: A compact and provably secure mix format.
\newblock In {\em IEEE S\&P}, 2009.

\bibitem{danezis2005}
George Danezis and Andrei Serjantov.
\newblock Statistical disclosure or intersection attacks on anonymity systems.
\newblock In {\em IH}, 2005.

\bibitem{das2017anonymity}
Debajyoti Das, Sebastian Meiser, Esfandiar Mohammadi, and Aniket Kate.
\newblock Anonymity trilemma: Strong anonymity, low bandwidth overhead, low
  latency—choose two.
\newblock In {\em IEEE S\&P}, 2018.

\bibitem{dms04}
Roger Dingledine, Nick Mathewson, and Paul Syverson.
\newblock {Tor}: The second-generation onion router.
\newblock In {\em USENIX Security}, 2004.

\bibitem{dingledine2009performance}
Roger Dingledine and Steven~J. Murdoch.
\newblock {Performance Improvements on Tor or, Why Tor is slow and what we're
  going to do about it}.
\newblock
  "\url{http://www.torproject.org/press/presskit/2009-03-11-performance.pdf},
  2009.
\newblock Retrieved on May. 23, 2016.

\bibitem{evans2009practical}
Nathan~S Evans, Roger Dingledine, and Christian Grothoff.
\newblock A practical congestion attack on tor using long paths.
\newblock In {\em USENIX Security}, 2009.

\bibitem{freedman2002tarzan}
Michael~J Freedman and Robert Morris.
\newblock Tarzan: A peer-to-peer anonymizing network layer.
\newblock In {\em ACM CCS}, 2002.

\bibitem{godfrey2009pathlet}
P~Godfrey, Igor Ganichev, Scott Shenker, and Ion Stoica.
\newblock Pathlet routing.
\newblock {\em ACM SIGCOMM CCR}, 39(4), 2009.

\bibitem{golle2004dining}
Philippe Golle and Ari Juels.
\newblock Dining cryptographers revisited.
\newblock In {\em Eurocrypt}, 2004.

\bibitem{Gong2011}
Xun Gong, Nikita Borisov, Negar Kiyavash, and Nabil Schear.
\newblock {Website Detection Using Remote Traffic Analysis}.
\newblock In {\em PETS}, 2012.

\bibitem{DBLP:journals/tissec/HopperVC10}
Nicholas Hopper, Eugene~Y. Vasserman, and Eric Chan{-}Tin.
\newblock How much anonymity does network latency leak?
\newblock {\em {ACM} TISSEC}, 2010.

\bibitem{houmansadr_swirl:_2011}
Amir Houmansadr and Nikita Borisov.
\newblock {SWIRL}: {A} {Scalable} {Watermark} to {Detect} {Correlated}
  {Network} {Flows}.
\newblock In {\em {NDSS}}, 2011.

\bibitem{houmansadr2013}
Amir Houmansadr and Nikita Borisov.
\newblock The need for flow fingerprints to link correlated network flows.
\newblock In {\em {PETS}}, 2013.

\bibitem{houmansadr_rainbow:_2009}
Amir Houmansadr, Negar Kiyavash, and Nikita Borisov.
\newblock {RAINBOW}: A robust and invisible non-blind watermark for network
  flows.
\newblock In {\em {NDSS}}, 2009.

\bibitem{Hsiao2012}
Hsu~Chun Hsiao, Tiffany Hyun~Jin Kim, Adrian Perrig, Akira Yamada, Samuel~C.
  Nelson, Marco Gruteser, and Wei Meng.
\newblock {LAP}: Lightweight anonymity and privacy.
\newblock In {\em IEEE S\&P}, 2012.

\bibitem{juarez2014}
Marc Juarez, Sadia Afroz, Gunes Acar, Claudia Diaz, and Rachel Greenstadt.
\newblock A critical evaluation of website fingerprinting attacks.
\newblock In {\em {ACM} {CCS}}, 2014.

\bibitem{herd2015}
Stevens Le~Blond, David Choffnes, William Caldwell, Peter Druschel, and
  Nicholas Merritt.
\newblock Herd: A scalable, traffic analysis resistant anonymity network for
  {VoIP} systems.
\newblock In {\em {ACM} {SIGCOMM}}, 2015.

\bibitem{aqua}
Stevens Le~Blond, David Choffnes, Wenxuan Zhou, Peter Druschel, Hitesh Ballani,
  and Paul Francis.
\newblock Towards {Efficient} {Traffic}-analysis {Resistant} {Anonymity}
  {Networks}.
\newblock In {\em {ACM} {SIGCOMM}}, 2013.

\bibitem{lee2017replay}
Taeho Lee, Christos Pappas, Adrian Perrig, Virgil Gligor, and Yih-Chun Hu.
\newblock The case for in-network replay suppression.
\newblock In {\em ACM AsiaCCS}, 2017.

\bibitem{levine2004timing}
Brian~N Levine, Michael~K Reiter, Chenxi Wang, and Matthew Wright.
\newblock Timing attacks in low-latency mix systems.
\newblock In {\em FC}. Springer, 2004.

\bibitem{timing-fc2004}
Brian~N. Levine, Michael~K. Reiter, Chenxi Wang, and Matthew~K. Wright.
\newblock Timing attacks in low-latency mix-based systems.
\newblock In {\em FC}, 2004.

\bibitem{mathewson2005}
Nick Mathewson and Roger Dingledine.
\newblock Practical traffic analysis: Extending and resisting statistical
  disclosure.
\newblock In {\em {PETS}}, 2005.

\bibitem{throughput2011}
Prateek Mittal, Ahmed Khurshid, Joshua Juen, Matthew Caesar, and Nikita
  Borisov.
\newblock Stealthy traffic analysis of low-latency anonymous communication
  using throughput fingerprinting.
\newblock In {\em {ACM} {CCS}}, 2011.

\bibitem{mittal2011stealthy}
Prateek Mittal, Ahmed Khurshid, Joshua Juen, Matthew Caesar, and Nikita
  Borisov.
\newblock Stealthy traffic analysis of low-latency anonymous communication
  using throughput fingerprinting.
\newblock In {\em ACM CCS}, 2011.

\bibitem{murdoch2006hot}
Steven~J Murdoch.
\newblock Hot or not: Revealing hidden services by their clock skew.
\newblock In {\em ACM CCS}, 2006.

\bibitem{murdoch2005low}
Steven~J. Murdoch and George Danezis.
\newblock {Low-cost traffic analysis of Tor}.
\newblock In {\em IEEE S\&P}, 2005.

\bibitem{murdoch2007sampled}
Steven~J Murdoch and Piotr Zieli{\'n}ski.
\newblock Sampled traffic analysis by internet-exchange-level adversaries.
\newblock In {\em PETS}, 2007.

\bibitem{overlier2006locating}
Lasse Overlier and Paul Syverson.
\newblock Locating hidden servers.
\newblock In {\em IEEE S\&P}, 2006.

\bibitem{pfitzmann2001}
Andreas Pfitzmann and Marit K\"{o}hntopp.
\newblock Anonymity, unobservability, and pseudonymity - a proposal for
  terminology.
\newblock In {\em PETS}, 2001.

\bibitem{piotrowska2017loopix}
Ania Piotrowska, Jamie Hayes, Tariq Elahi, Sebastian Meiser, and George
  Danezis.
\newblock The loopix anonymity system.
\newblock In {\em Usenix Security}, 2017.

\bibitem{pries2008new}
Ryan Pries, Wei Yu, Xinwen Fu, and Wei Zhao.
\newblock A new replay attack against anonymous communication networks.
\newblock In {\em ICC'08}, pages 1578--1582. IEEE, 2008.

\bibitem{Putze2007}
Felix Putze, Peter Sanders, and Johannes Singler.
\newblock Cache-, hash- and space-efficient bloom filters.
\newblock In {\em Workshop on Experimental Algorithms}, 2007.

\bibitem{Sankey2014}
Jody Sankey and Matthew Wright.
\newblock {Dovetail: Stronger anonymity in next-generation internet routing}.
\newblock In {\em PETS}, 2014.

\bibitem{sherwood2002p}
Rob Sherwood, Bobby Bhattacharjee, and Aravind Srinivasan.
\newblock P$^5$: A protocol for scalable anonymous communication.
\newblock In {\em IEEE S\& P}, 2002.

\bibitem{shmatikov2006timing}
Vitaly Shmatikov and Ming-Hsiu Wang.
\newblock Timing analysis in low-latency mix networks: Attacks and defenses.
\newblock In {\em ESORICS}. Springer, 2006.

\bibitem{sun2015raptor}
Yixin Sun, Anne Edmundson, Laurent Vanbever, Oscar Li, Jennifer Rexford, Mung
  Chiang, and Prateek Mittal.
\newblock Raptor: routing attacks on privacy in tor.
\newblock In {\em USENIX Security 15}, 2015.

\bibitem{Sundaresan2011}
Srikanth Sundaresan, Walter de~Donato, Nick Feamster, Renata Teixeira, Sam
  Crawford, and Antonio Pescap{\`{e}}.
\newblock {Broadband Internet performance: a view from the gateway}.
\newblock In {\em ACM SIGCOMM}, 2011.

\bibitem{syverson2001towards}
Paul Syverson, Gene Tsudik, Michael Reed, and Carl Landwehr.
\newblock Towards an analysis of onion routing security.
\newblock In {\em PETS}, 2001.

\bibitem{vanbever2014anonymity}
Laurent Vanbever, Oscar Li, Jennifer Rexford, and Prateek Mittal.
\newblock Anonymity on quicksand: Using bgp to compromise tor.
\newblock In {\em ACM HotNet}, 2014.

\bibitem{Wang2014}
Tao Wang, Xiang Cai, Rishab Nithyanand, Rob Johnson, and Ian Goldberg.
\newblock Effective attacks and provable defenses for website fingerprinting.
\newblock In {\em USENIX Security}, 2014.

\bibitem{wang2008dependent}
Wei Wang, Mehul Motani, and Vikram Srinivasan.
\newblock Dependent link padding algorithms for low latency anonymity systems.
\newblock In {\em CCS}. ACM, 2008.

\bibitem{wang2003}
Xinyuan Wang and Douglas~S. Reeves.
\newblock Robust correlation of encrypted attack traffic through stepping
  stones by manipulation of interpacket delays.
\newblock In {\em {ACM} {CCS}}, 2003.

\bibitem{wolinsky2012dissent}
David~Isaac Wolinsky, Henry Corrigan-Gibbs, Bryan Ford, and Aaron Johnson.
\newblock Dissent in numbers: Making strong anonymity scale.
\newblock In {\em Usenix OSDI}, 2012.

\bibitem{wolinsky2013}
David~Isaac Wolinsky, Ewa Syta, and Bryan Ford.
\newblock Hang with your buddies to resist intersection attacks.
\newblock In {\em {ACM} {CCS}}, 2013.

\bibitem{Yang2007NIRA}
Xiaowei Yang, David Clark, and Arthur~W Berger.
\newblock {NIRA}: a new inter-domain routing architecture.
\newblock {\em IEEE/ACM TON}, 15(4), 2007.

\bibitem{Xin2011SCION}
Xin Zhang, Hsu-Chun Hsiao, Geoffrey Hasker, Haowen Chan, Adrian Perrig, and
  David~G. Andersen.
\newblock {SCION}: Scalability, control, and isolation on next-generation
  networks.
\newblock In {\em IEEE S\&P}, 2011.

\bibitem{zhang2000detecting}
Yin Zhang and Vern Paxson.
\newblock Detecting stepping stones.
\newblock In {\em USENIX Security}, 2000.

\bibitem{Zhu2005}
Ye~Zhu, Xinwen Fu, Bryan Graham, Riccardo Bettati, and Wei Zhao.
\newblock {On flow correlation attacks and countermeasures in mix networks}.
\newblock In {\em PET}, 2004.

\end{thebibliography}
}

\end{document}